\documentclass[10pt]{article}

\usepackage{amssymb} 
\usepackage{amsmath, amsthm}
\usepackage{amsfonts}

\usepackage[utf8]{inputenc}
\usepackage{graphicx} 
\usepackage{color}

\usepackage{verbatim}

\definecolor{blue1}{rgb}{0.1,0.3,0.6}
\definecolor{blue2}{rgb}{0.15,0.35,0.65}
\definecolor{blue3}{rgb}{0.2,0.4,0.7}
\definecolor{blue4}{rgb}{0.3,0.5,0.8}
\definecolor{darkgrey1}{rgb}{0.05,0.05,0.05}
\definecolor{grey1}{rgb}{0.45,0.45,0.45}
\definecolor{red1}{rgb}{0.85,0.1,0.1}

\newcommand{\LL}{\mathcal{L}}
\newcommand{\rhobar}{\overline{\rho}}

\let\oldbibliography\thebibliography
\renewcommand{\thebibliography}[1]{%
  \oldbibliography{#1}%
  \setlength{\itemsep}{0pt}%
}

\usepackage{setspace}
\usepackage{titlesec}
\usepackage[font=small,labelfont=bf]{caption}

\usepackage{authblk}

\newtheorem{proposition}{Proposition}[section]

\begin{document}

\color{darkgrey1}

\title{Acceleration of convergence to equilibrium in Markov chains by breaking detailed balance%\thanks{Grants or other notes
}
%
%\titlerunning{Accelerated convergence by breaking detailed balance}        % if too long for running head

\author[1]{Marcus Kaiser}
\author[2]{Robert L. Jack}
\author[1]{Johannes Zimmer}
\affil[1]{Department of Mathematical Sciences, University of Bath}
\affil[2]{Department of Physics, University of Bath}

\maketitle

%% begin paste

\begin{abstract}
  We analyse and interpret the effects of breaking detailed balance on the convergence to equilibrium of conservative interacting particle
  systems and their hydrodynamic scaling limits.  For finite systems of interacting particles, we review existing results showing
  that irreversible processes converge faster to their steady state than reversible ones.  We show how this behaviour appears in
  the hydrodynamic limit of such processes, as described by macroscopic fluctuation theory, and we provide a quantitative
  expression for the acceleration of convergence in this setting.  We give a geometrical interpretation of this acceleration, in
  terms of currents that are \emph{antisymmetric} under time-reversal and orthogonal to the free energy gradient, which act to
  drive the system away from states where (reversible) gradient-descent dynamics result in slow convergence to equilibrium.
\end{abstract}

\section{Introduction}
\label{sec:Introduction}

{In this paper we analyse the effects of breaking detailed balance for interacting particle systems (as described by Markov
processes~\cite{Liggett2005a}), and their hydrodynamic scaling limits (as described by Macroscopic Fluctuation Theory~\cite{Bertini2015a}).}
%Hence, we consider two families of model systems: the interacting particle
%systems~\cite{Liggett2005a} are described by Markov chains, and their associated hydrodynamic theories fall within the framework of macroscopic
%fluctuation theory~\cite{Bertini2015a}.}  
The interacting particle systems represent microscopic descriptions of physical systems,
in which the motion of each particle may be followed individually.  The (fluctuating) hydrodynamic model of the same system
describes its behaviour on large length and time scales, in which case the motion of the individual particles is no longer
visible, and one works instead with a smooth density field, whose time evolution includes a deterministic element as well as a
(weak) stochastic noise~\cite{Kipnis1999a}.

Among interacting particle systems, those with detailed balance are special -- they correspond to Markov chains that are
reversible with respect to an invariant measure $\pi$.  Physically, these models are important because their steady states are
time-reversal symmetric and lack any persistent currents, so they can be used to describe systems that relax to states of thermal
equilibrium.  They also have applications outside physics, because given a (possibly non-normalised) measure $\nu$, it is
straightforward to design a reversible Markov chain whose invariant measure $\pi$ is proportional to $\nu$.  This construction is
at the root of many Markov chain Monte Carlo (MCMC) methods~\cite{Asmussen2007a,Newman1999a}, in which one typically aims to
generate large numbers of uncorrelated samples from a prescribed distribution $\pi$.  Such methods have widespread applications
including Bayesian learning, protein folding and cryptography~\cite{Diaconis2009a}.

In both the physical systems and the MCMC methods, an important question is the rate of convergence to equilibrium of the relevant
Markov chains.  In MCMC, this rate controls the computational cost required to obtain independent samples from $\pi$, which is an
important factor in the efficiency of the method.  In the physical systems, the question of how fast a system converges to
equilibrium controls many physical properties including fluid viscosities, and systems' abilities to response to changes in
external conditions, such as temperature.

Recently, several results have become available which show that for a given invariant measure $\pi$, reversible Markov chains have
the slowest convergence~\cite{Hwang2005a,Sun2010a,Lelievre2013a,Bierkens2015a,Rey-Bellet2015a,Rey-Bellet2016a}.  Given that most
common MCMC methods are based on such reversible models, and that faster convergence is linked to improved efficiency, this
observation offers a route towards the development of new and more efficient methods, some of which are already becoming
available~\cite{Bernard2009a}. Breaking reversibility can be achieved by an explicit modification of transition
rates~\cite{Rey-Bellet2016a}, or by an expansion of the state space (\emph{lifting}) to incorporate persistence of motion or
inertial effects~\cite{Chen1999a,Diaconis2000a}. The main physical feature of the resulting irreversible Markov chains is that
they (generically) have non-equilibrium steady states characterised by finite entropy production and dissipation energy.  Compared
to the equilibrium setting, the nature of fluctuations and convergence to steady state in non-equilibrium systems is much less
understood, and is an area of important current activity~\cite{Derrida2007a,Baiesi2009a,Bertini2015a}.

To address these questions, this paper presents several new results.  First, we revisit existing results for microscopic models,
concentrating in particular on the spectral gap of the generator, and how it is affected when detailed balance is broken.  Second,
we investigate how breaking detailed balance affects the hydrodynamic limit of the model -- in this latter case, convergence to
equilibrium is most easily analysed via large deviations of the empirical measure~\cite{Rey-Bellet2015a,Rey-Bellet2016a}.  Third,
we illustrate our general results by numerical results of a simple interacting particle system -- the zero-range
process~\cite{Spitzer1970a}.  These numerical results are particularly relevant since the analytical results indicate that
breaking detailed balance can never slow down convergence to equilibrium, but they provide rather little insight into how much
this convergence can be accelerated, nor how this effect depends on the specific way in which detailed balance is broken.  We
provide some general remarks and comments in this direction.

\subsection{Characterisation of convergence to steady state}
\label{sec:Char-conv-to}

A number of methods are available to analyse the time required for a system to reach its steady state.  This section contains a
brief review of some of them.  For microscopic models -- such as Markov processes on (finite) discrete spaces and SDEs -- we
mention some recent work showing how breaking detailed balance can accelerate convergence of systems to their steady states.
These results serve as a foundation for our results here, which show how these effects manifest on the macroscopic scale.

\subsubsection{Spectral gap}
\label{sec:Spectral-gap}

The first -- and most common -- method for analysis of convergence to equilibrium is to estimate the spectral gap of the generator
of the relevant stochastic process.  {In general, the eigenvalues $\{\lambda_i\}$ of the generator are complex numbers, there is a simple eigenvalue $\lambda_0=0$ and all other eigenvalues have negative real parts.  The spectral gap $\alpha_{\rm min}$ is the minimal value of $|\lambda_i^r|$ among the non-zero eigenvalues, where $\lambda_i^r$ denotes the real part of $\lambda_i$.
%
%{Considering the eigenvalue with largest non-zero real part, the spectral gap is
%  defined as the magnitude of the real part of this eigenvalue.}
%
 Roughly speaking, the physical significance of the spectral gap} is that the system converges
exponentially fast to its steady state, with a characteristic time scale
\begin{equation}
  \label{equ:def-taug}
  \tau_{\rm g}=(1/\alpha_{\rm min}).
\end{equation}  
For stochastic differential equations~\cite{Hwang2005a,Lelievre2013a} and discrete-space Markov processes~\cite{Rey-Bellet2016a},
it has been shown that irreversible processes generically have smaller time scales $\tau_{\rm g}$, compared to reversible
processes with the same invariant measure.  We provide further results in this direction in Sec.~\ref{sec:spec-gap} below, for the
discrete space Markov processes that are relevant for interacting particle systems.

\subsubsection{Asymptotic variance}
\label{sec:Asymptotic-variance}

Another set of methods for the analysis of the convergence to steady state is based on empirical time averages.  That is, let
$X_t$ be the state of the system at time $t$ and let $f$ be an observable quantity (test function) whose value at time $t$ is
$f(X_t)$.  Then the empirical time average of $f$ is
\begin{equation}
  \label{eqn:f_avg}
  \overline f(T) :=  \frac1T \int_0^T f(X_s) \;\! ds . 
\end{equation}
The quantity $\overline f(T)$ is a random variable which -- under suitable conditions related to ergodicity -- converges almost
surely to the expectation value of $f$, which we denote by $\mathbb{E}_\pi(f)$.

{Moreover the distribution of $\sqrt T (\bar f(T)-\mathbb E_\pi(f))$ converges by the central limit theorem to a normal distribution with variance  $\sigma^2_f$.
The latter is referred to as asymptotic variance or time average variance constant (TAVC) which can be obtained as $\sigma^2_f=\lim_{T\to\infty} T\mathrm{Var}(\bar{f}(T))$, see in \cite[Chapter IV]{Asmussen2007a}, \cite{Rey-Bellet2015a} and \cite[Section 3.5]{Touchette2009a}.
Hence, the variance of $\overline{f}(T)$ decays for large times as
$\mathrm{Var}(\overline{f}(T))\sim \sigma^2_f/T$. It is then natural to identify a
time scale $\tau^f_{\rm v}:=\sigma^2_f/\mathrm{Var}_\pi(f)$.}
Note that $\tau^f_{\rm v}$ depends on the observable $f$ of interest --
roughly speaking it represents the autocorrelation time of $f(X_t)$.  In general $\tau^f_{\rm v}$ and $\tau_{\rm g}$ are different
time scales: $\tau^f_{\rm v}$ controls the convergence of $\overline f(T)$ while $\tau_{\rm g}$ controls the convergence of the
probability measure itself.  As for $\tau_{\rm g}$, one finds that $\sigma^2_f$ can be reduced by breaking detailed balance in Markov
chains~\cite{Bierkens2015a,Sun2010a} and SDEs~\cite{Duncan2016a,Hwang2015a}.

\subsubsection{Large deviations at level-1 and level-2}
\label{sec:largedev-intro}

A more detailed analysis of the large-$T$ behaviour of $\overline f(T) $ is available from large deviation
theory~\cite{Hollander2000a,Touchette2009a}.  Informally, one expects that for large $T$, the random variable $\overline{f}(T)$
satisfies
\begin{equation} 
  \label{equ:ldp-f} 
  \mathrm{Prob}\Bigl[ \overline f(T) \approx \hat f\Bigr] \asymp {\rm e}^{-T I_f(\hat f)} 
\end{equation} 
for some \emph{rate function} $I_f$ (which depends on the choice of test function $f$).  We use the notation in~\eqref{equ:ldp-f}
throughout this work as an informal way to state large deviation principles: it means that the log probability that
$\overline f(T)$ takes a value in a small interval containing $\hat{f}$ can be bounded above and below by quantities related to
the rate function $I_f$~\cite{Hollander2000a,Touchette2009a}.  The rate function achieves its minimal value of zero when $\hat f$
is equal to $\mathbb{E}_\pi(f)$, and the second derivative of $I_f$ at this minimum is related to $\sigma^2_f$.  The function $I_f$ is
a level-1 rate function~\cite{Hollander2000a}.

A yet more detailed analysis is available by considering not just the large deviations of a single test function $f$ but instead
to consider large deviations of the empirical measure.  That is, for a Markov chain on a discrete space $\Omega$, define the empirical measure
{\begin{equation}\label{eqn:emprical_avg}
  \bar\mu_T(x) :=  \frac1T \int_0^T \delta_{X_s,x} \;\! ds , 
\end{equation}}
where $\delta_{x,y}$ is a Kronecker delta. The empirical measure at time $t$ is a vector $\bar\mu_T = (\bar\mu_T(x))_{x\in\Omega}$. For
large enough $T$, ergodicity implies that $\bar\mu_T(x)$ converges almost surely to $\pi(x)$, and the fluctuations of the measure
$\mu$ in this limit are described by a large deviation principle at level-2:
\begin{equation}
  \mathrm{Prob}\bigl[\bar\mu_T \approx \nu\bigr] \asymp {\rm e}^{-T I_2(\nu)}
\end{equation}
where $I_2$ is the rate function~\cite{Donsker1983a}, which now depends on a vector $\nu$ instead of a single real argument
$\hat{f}$.  Note that the (level-1) rate function $I_f$ for any observable $f$ can be obtained by a contraction of this large
deviation principle, so the function $I_2$ contains a great deal of information about the convergence of a system to its steady
state.  Moreover, as might be expected from the terminology ``rate function'', the quantity $1/I_2(\mu)$ has an interpretation as
a $\mu$-dependent time scale associated with the decay of an initial measure $\mu$ to the invariant measure $\pi$.

Recent work by Rey-Bellet and Spiliopoulos~\cite{Rey-Bellet2015a,Rey-Bellet2016a} has motivated the analysis of $I_2$ as a measure
of the rate of convergence of processes to their steady states.  Their work, and that of Bierkens~\cite{Bierkens2015a}, show that
breaking detailed balance accelerates this convergence.  Note however that in contrast to the spectral gap -- where a single
number characterises the rate of convergence of the whole system -- the rate function $I_2$ depends on the measure $\mu$ for which
it is evaluated; similarly the asymptotic variance $\sigma^2_f$ depends on the specific observable $f$.  In this sense, the
information available from the asymptotic variance and the large deviations is greater than that available from the spectral gap,
but this extra information may also make these measures harder to interpret in terms of simple acceleration or slowing down of
convergence to equilibrium.  {Of course, other useful measurements of convergence rates are available, such as mixing
  times~\cite{Levin2009a}, cutoff phenomena (see e.g.~\cite{Labbe2016}, where cutoff was recently established for the asymmetric
  simple exclusion process) and log-Sobolev constants (e.g.~\cite{Diaconis1996a}), but these are not analysed in this work.}

\subsection{Outline}
\label{sec:Outline}

The remainder of this paper is organised as follows: Section~\ref{sec:theory} includes a theoretical analysis of the effects of
breaking detailed balance on convergence to steady states, including both Markov chains (Sec.~\ref{sec:markov-theory}) and
hydrodynamic limits (Sec.~\ref{sec:macro-theory}).  Section~\ref{sec:numerics} presents numerical results that illustrate this
acceleration in the zero-range process: we provides examples in both one-dimensional and two-dimensional settings.  Finally,
Section~\ref{sec:conc} contains our conclusions.

\section{Theoretical results}
\label{sec:theory}

\subsection{Acceleration of the microscopic dynamics}
\label{sec:markov-theory}

In this section, we consider an irreducible Markov jump process on a finite state space $\Omega$ which contains $n$ states.  In
terms of interacting particle systems, this process describes the dynamics of a finite number of particles that move on some
finite lattice.  The process is defined by rates $c(x\to y)$ for states $x,y\in\Omega$.  The condition of \emph{detailed balance}
(or \emph{reversibility}) is that for some probability measure $\pi$ and all $x,y$ then
\begin{equation}
  \label{eqn:DB}
  \pi(x) c(x\to y) = \pi(y)c(y\to x) .
\end{equation}
In this case the (unique) invariant measure of the Markov process is $\pi$.

Let the generator of the Markov process be $\LL$.  The generator has a representation as an $n\times n$ matrix and the
reversibility condition~\eqref{eqn:DB} corresponds to symmetry of $\LL$ with respect to the $L^2(\pi)$ inner product
$\langle f,g\rangle_{\pi}=\sum_x f(x) g(x) \pi(x)$. If detailed balance is broken (non-reversible Markov chain), then $\LL$ is not
symmetric with respect to $L^2(\pi)$, but one may always write
\begin{equation}
  \label{equ:LSA}
  \LL = \LL_S + \LL_A, 
\end{equation}
where $\LL_S$ is symmetric with respect to $L^2(\pi)$, while $\LL_A$ is antisymmetric.  Moreover, $\LL_S$ is a generator for a
reversible stochastic process, whose transition rates may be verified to be
\begin{equation}
\label{equ:cs}
  c_s(x\to y) = \frac12 \left[ c(x\to y) + \pi(y) c(y\to x) \pi(x)^{-1} \right] , 
\end{equation}
where $\pi$ is the invariant measure of $\LL$, which is also the invariant measure of $\LL_S$.  (Recall that the original Markov
process is finite and irreducible, which ensures that $\pi(x)>0$ for all $x$).  We also identify the off-diagonal elements of
${\cal L}_A$ as
\begin{equation*} 
  c_a(x\to y):=c(x\to y)-c_s(x\to y). 
\end{equation*}  
Hence one has from~\eqref{equ:cs} that
\begin{equation}
  \label{equ:cadj}
  \pi(x)[c_s(x\to y) + c_a(x\to y)]=\pi(y)[c_s(y\to x) - c_a(y\to x)] .
\end{equation}
Note that $\mathcal L$ and $\mathcal L_S$ both are generators, whereas the operator $\mathcal L_A$ is not a generator of a Markov chain.

{Alternatively one can think of the decomposition of $\mathcal L$ in $\mathcal L_S$ and $\mathcal L_A$ as follows:
  Consider the Markov process $\eta_t$ (with $t\in[-T,T]$ for some $T>0$) associated to $\mathcal L$ distributed according to the
  steady state $\pi$. The time reversed process $\hat\eta(t):=\eta(-t)$ is also associated to a generator, $\mathcal L^\ast$, say. The
  symmetric part of the generator can be recovered as $\mathcal L_S = (\mathcal L + \mathcal L^\ast)/2$.  }

Given these preliminaries, we can now be precise about the sense in which breaking detailed balance accelerates convergence: in
all cases we compare the process $\LL$ with the corresponding symmetrised process $\LL_S$. (Equivalently, one may imagine starting
from a reversible process $\LL_S$ and breaking detailed balance by adding an extra term $\LL_A$ to the generator.)  The processes
$\LL$ and $\LL_S$ both converge to the same invariant measure $\pi$ --- one aims to prove that convergence times such as
$\tau_{\rm g}$ or $1/I(\mu)$ are smaller for $\LL$ than for $\LL_S$.

\subsubsection{The spectral gap}
\label{sec:spec-gap}

To illustrate how breaking detailed balance accelerates convergence, we show in Prop.~\ref{prop:spectralgap} below that breaking
detailed balance can only increase the spectral gap, so that the convergence of the irreversible process is characterised by a
smaller value of the time $\tau_{\rm g}$.  This result has been proven in greater generality in~\cite{Hwang2005a,Rey-Bellet2016a},
but we provide a short proof here, for illustrative purposes.

{ To this end, consider an initial measure $\mu_0$, and represent it in terms of an eigendecomposition of $\LL$, so
  that
\begin{equation}
  \mu_0(x) = \pi(x) + \sum_{j=1}^m \bigl(\alpha_j\nu_j(x) + \bar\alpha_j \bar\nu_j(x)\bigr), 
\end{equation}
where the $\alpha_j\in\mathbb C$ are $\mu_0$-dependent coefficients, while $\nu_j$ are complex-valued measures which are left-eigenvectors corresponding to eigenvalues $\lambda_j$ of $\LL$. The overbar (e.g. $\bar\alpha$) denotes the complex conjugate.  Decomposing the non-zero eigenvalues $\lambda_j$ into real and imaginary parts, as $\lambda_j=\lambda_j^r+i\lambda_j^i$, the measure at time $t$ is given by
\begin{equation}\label{eqn:pdf_density}
  \mu_t = (\mathrm e^{t\mathcal{L}})^\dagger\mu_0 = \pi + \sum_{j=1}^m \mathrm e^{\lambda_j^rt}\bigl(\mathrm e^{i\lambda_j^i t}\alpha_j\nu_j 
  +  \mathrm e^{-i\lambda_j^i t}\bar\alpha_j\bar\nu_j\bigr), 
\end{equation}
where $\cdot^\dagger$ denotes a matrix transpose. Note that for real-valued eigenvalues (with $\lambda_j^i=0$) the term in brackets is equal to $2\alpha_j \nu_j$, as in this case also the left (and right) eigenvectors are real-valued.

 Moreover, $\lambda_j^r<0$ for all $j$, since $\LL$ is the generator of an  irreducible finite Markov process. }
One sees immediately that this Markov process relaxes exponentially fast to its steady state.  Moreover, the rate of this exponential decay is controlled by the non-zero eigenvalue of $\LL$ whose real part is smallest in magnitude. Similar results to the following proposition have already been obtained in e.g. \cite{ichiki2013,sakai2016}:

\begin{proposition}
  \label{prop:spectralgap} 
  Let $\mathcal L$ and $\mathcal L_S$ be given as above.  The non-zero eigenvalues of $-\LL_S$ are real and positive; let the
  smallest such eigenvalue be $\alpha_{\rm min}$ and the largest be $\alpha_{\rm max}$.  Then every non-zero eigenvalue $\lambda$ of $-\LL$
  satisfies 
\begin{equation}
\alpha_{\rm min} \leq {\rm Re}(\lambda) \leq \alpha_{\rm max}.
  \label{equ:alpha-bounds}
\end{equation} 
\end{proposition}

\begin{proof}
  Define the Dirichlet form for $\LL$ as
  $\mathcal E(f,g):= \langle f,-\mathcal L g\rangle_\pi = \sum_x f(x) \mathcal (-\mathcal Lg)(x)\;\!\pi(x)$, where $\pi$ is the
  unique stationary distribution of $\mathcal L$. {Let $\lambda$ be a non-zero eigenvalue of $-\mathcal L$ with corresponding right eigenvector $f_\lambda+ig_\lambda$. As $\mathcal E(1,f)=0$ for all $f$, we obtain    $0=\mathcal E(1,f_\lambda + i g_\lambda) = \lambda (\langle 1,f_\lambda\rangle_\pi + i \langle 1,g_\lambda \rangle_\pi)$. Since
    $\lambda$ is non-zero, we obtain that $\langle 1,f_\lambda\rangle_\pi = 0 = \langle 1,g_\lambda \rangle_\pi$.  This implies that both $f_\lambda$ and $g_\lambda$ are mean zero, so
    $\mathrm {Var}_\pi(h) = \langle h,h\rangle_\pi$ for $h\in\{f_\lambda,g_\lambda\}$.  Since
    $\mathcal E(f_\lambda-ig_\lambda, f_\lambda+ig_\lambda)= \lambda \langle f_\lambda-ig_\lambda,f_\lambda+ig_\lambda\rangle_\pi
    = \lambda (\langle f_\lambda,f_\lambda\rangle_\pi + \langle g_\lambda,g_\lambda\rangle_\pi)$,
    the bilinearity of the Dirichlet form yields that the real part of $\lambda$ is given by}
\begin{equation}
  \label{equ:reim}
  \operatorname{Re}(\lambda) = \frac{\mathcal E(f_\lambda,f_\lambda) + \mathcal E(g_\lambda,g_\lambda)}
  {\langle f_\lambda,f_\lambda\rangle_\pi + \langle g_\lambda,g_\lambda\rangle_\pi}.
\end{equation}
In addition, one has { (for the min and max taken over the two cases $h=f_\lambda$ and $h=g_\lambda$)}
\begin{equation}
  \label{equ:hh}
  \min_{h\in \{f_\lambda,g_\lambda\}}\frac{\mathcal E(h,h)}{\langle h,h\rangle_\pi}
  \le\frac{\mathcal E(f_\lambda,f_\lambda) + \mathcal E(g_\lambda,g_\lambda)}{\langle f_\lambda,f_\lambda\rangle_\pi 
    + \langle g_\lambda,g_\lambda\rangle_\pi}
  \le\max_{h\in \{f_\lambda,g_\lambda\}}\frac{\mathcal E(h,h)}{\langle h,h\rangle_\pi}.
\end{equation}
Define ${\cal E}_S(f,g)=\langle f,-\mathcal L_S g\rangle_\pi$, and note that ${\cal E}(h,h)={\cal E}_S(h,h)$.  {Also
  $\alpha_{\rm min}=\mathrm{min}_{h: \langle 1,h\rangle_\pi = 0} \frac{\mathcal E_S(h,h)}{\langle h,h\rangle_\pi}$ }.  Hence the
left hand side of~\eqref{equ:hh} is bounded below by $\alpha_{\rm min}$.  Applying a similar argument to the right hand side
of~\eqref{equ:hh} and combining with~\eqref{equ:reim} finally yields {(\ref{equ:alpha-bounds})}.
%\begin{equation} 
%  \alpha_{\min}\le \operatorname{Re}(\lambda) \le \alpha_{\max} . 
%\end{equation}
\end{proof}

\subsubsection{Bounds on level-2 rate functions for discrete Markov processes}

From Prop.~\ref{prop:spectralgap} and using~\eqref{equ:def-taug}, one clearly has
\begin{equation} 
  \tau_{\rm g}^{\rm irr} \leq \tau^{\rm rev}_{\rm g}. 
\end{equation}
That is, the irreversible process converges to its steady state at least as quickly as the reversible one.  A similar
argument~\cite{Bierkens2015a} establishes that the level-2 rate functions for $\LL$ and $\LL_S$ are related as
\begin{equation}
  \label{equ:i2-faster}
  I_2(\mu) \ge I_2^S(\mu) , 
\end{equation}
again establishing a faster rate of convergence on breaking detailed balance. { Recall that results of the form
  \eqref{equ:i2-faster} yield information about the empirical measure $\bar\mu_T$ defined in~\eqref{eqn:emprical_avg}, whereas the
  {previous} result {(\ref{equ:alpha-bounds})} concerns the spectral gap and the convergence of $\mu_t$, the distribution of the process at time $t$ as defined in
  \eqref{eqn:pdf_density}. Note that {$\bar\mu_T$} is a random quantity, whereas {$\mu_t$} is the solution to a deterministic
  differential equation.}

We now show (Prop.~\ref{prop:ldpbounds}) that the rate of convergence of the irreversible model has an upper bound, as well as the
lower bound given by $I_2^S(\mu)$.  That is, $I_2(\mu)$ is bounded both above and below, just as the spectral gap is bounded
in~\eqref{equ:alpha-bounds}.  This limits the acceleration that is available by breaking detailed balance for (finite) discrete
Markov processes, in contrast to the situation for diffusions~\cite{Rey-Bellet2015a}.  {The proof for the following
  proposition is based on the variational formula for the level-2 LDP~\cite{Hollander2000a}. Whilst the lower bound, which is
  known in the literature, see e.g.~\cite{Bierkens2015a,Rey-Bellet2016a}, follows from the variational representation of the rate
  function, the upper bound is (to our knowledge) a novel result.}

\begin{proposition}
  \label{prop:ldpbounds}
  Consider a finite-state continuous-time Markov chain with generator $\mathcal L = \mathcal L_S +\mathcal L_A$ and transition
  rates $c(x\to y)=c_s(x\to y)+c_a(x\to y)$, as defined in Sec.~\ref{sec:markov-theory}.  The level-2 rate functional $I_2(\mu)$
  is bounded as follows:
  \begin{multline}
    I_2^S(\mu)
    \le
    I_2(\mu)  
    \le
    I_2^S(\mu)
    + \sum_{x \not = y}\bigl[c_s(x\to y) - \sqrt{c_s(x\to y)^2-c_a(x\to y)^2}\bigr] 
    \\ \times \sqrt{\tfrac{\mu(x)}{\pi(x)}\tfrac{\mu(y)}{\pi(y)}}\pi(x),
    \label{equ:i2-bound}
  \end{multline}
  where the rate functional $I_2^S(\mu)$ for the reversible process with generator $\mathcal L_S$ is given by
  $I_2^S(\mu) = \bigl\langle \sqrt{\tfrac \mu\pi},- \mathcal L_S \sqrt{\tfrac \mu\pi}\bigr\rangle_\pi$.
\end{proposition}

\begin{proof}
The rate functional is given by a variational formula~\cite{Hollander2000a}:
\begin{equation*}
  I_2(\mu) 
= \sup_{f>0}\;\! \langle f^{-1}, -\mathcal Lf\rangle_\mu.
\end{equation*}
In the symmetric case ($\mathcal L=\mathcal L_S$) the maximum is $I_2^S(\mu)$, which is attained when $f=\sqrt{\mu/\pi}$. In
general we write $f=\sqrt{\mu/\pi}\;\!\mathrm e^{V}$ for some potential $V$.

A direct computation yields
\begin{equation}
  I_2^S(\mu) = \sum_{x\neq y} \left( \sqrt\frac{\mu(x)}{\pi(x)}-\sqrt\frac{\mu(y)}{\pi(y)}   \right)  \sqrt{\mu(x)\pi(x)}  c(x\to y)
\end{equation}
and
\begin{equation}
  I_2(\mu)
  = I_2^S(\mu) + \sup_V I_A(\mu,V) 
  \label{equ:supIA}
\end{equation}
with
\begin{equation}
  \label{eqn:maximise}
  I_A(\mu,V)
  =\sum_{x\neq y}\sqrt\frac{\mu(y)}{\pi(y)}\left( 1-\mathrm e^{V(y)-V(x)}\right)\sqrt{\mu(x)\pi(x)} c(x\to y).
\end{equation}
If $V$ is a constant function, then $I_A(\mu,V)=0$ so clearly $\sup_V I_A(\mu,V)\ge 0$. Hence, \eqref{equ:supIA} yields the lower
bound in~\eqref{equ:i2-bound}, as in~\cite{Bierkens2015a}.

For the upper bound, it is convenient to define $m(x,y):=\frac12\sqrt{\frac{\mu(x)\mu(y)}{\pi(x)\pi(y)}}$ and
$q(x,y):=\pi(x)c(x\to y)$. This yields
\begin{equation}
  \label{equ:ia-sym}
  I_A(\mu,V)
  = \sum_{x\not=y} m(x,y) \left[ (1-\mathrm e^{V(y)-V(x)})q(x,y) + (1-\mathrm e^{V(x)-V(y)})q(y,x)\right] , 
\end{equation}
where we have symmetrised the summand with respect to $x,y$.  For positive constants $a, b$, one may easily establish the general
inequality $a{\rm e}^{V} + b{\rm e}^{-V} \geq 2\sqrt{ab}$.  Applying this inequality to the summand in~\eqref{equ:ia-sym} yields
\begin{equation}
  \label{equ:ia-qq}
  I_A(\mu,V)
  \leq \sum_{x\not=y} m(x,y) \left[ q(x,y)+q(y,x) - 2\sqrt{q(x,y)q(y,x)}  \right] .
\end{equation}
From~\eqref{equ:cs},~\eqref{equ:cadj} one has $q(x,y)+q(y,x)=2c_s(x\to y)\pi(x)$ and
$q(x,y)q(y,x)=[c_s^2(x\to y)-c_a^2(x\to y)]\pi(x)^2$; substituting these results into~\eqref{equ:ia-qq} yields
\begin{equation*}
  I_A(\mu,V) \leq  \sum_{x\not=y}\sqrt{\tfrac{\mu(x)}{\pi(x)}\tfrac{\mu(y)}{\pi(y)}} \bigl[c_s(x\to y) 
  - \sqrt{c_s(x\to y)^2-c_a(x\to y)^2}\bigr]\pi(x),
\end{equation*}
and the combination with~\eqref{equ:supIA} establishes the upper bound in~\eqref{equ:i2-bound}.
\end{proof}

\subsubsection{Discussion}
\label{sec:markov-discuss}

Our intuition for the (bounded) acceleration by breaking detailed balance is as follows: for reversible processes we can think of $\mu_t$ (the distribution of the process at time $t$) undergoing a steepest descent process (gradient flow) for the free energy
$F(t)=\sum_x \mu_t(x) \log( \mu_t(x)/\pi(x))$, within a particular geometric setting~\cite{Maas2011a}. The precise nature of
this geometry is immaterial for this discussion: the key point is that relaxation to equilibrium is fast when the free energy
gradient is steep, and tends to be slow when it is shallow.  On breaking detailed balance, the free energy still decreases
monotonically, but its motion is no longer restricted to the direction of steepest descent.  This can have several possible
effects and the rate of change of $F(t)$ may either increase or decrease on breaking detailed balance. However, we argue that an
important contribution to the acceleration of convergence arises because the irreversible component of the dynamics drives the
system away from regions where the free energy gradient is shallow and into regions where it is steeper.  We will demonstrate this
effect explicitly at the hydrodynamic level, in Sec.~\ref{sec:hydro-accel}.

Notice however, that while slow processes associated with $\LL_S$ are accelerated by breaking detailed balance, the inequality
involving $\alpha_{\rm max}$ in Prop.~\ref{prop:spectralgap} implies that fast aspects of the relaxation tend to be slowed down.
Indeed, $\mathrm{tr}(\LL_A)=0$ so $\mathrm{tr}(\LL)=\mathrm{tr}(\LL_S)$: since the trace is equal to the sum of the eigenvalues,
one sees that if some (slow) processes are accelerated by breaking detailed balance another set of (faster) processes must be
slowed down by a similar amount.  Within the intuitive picture, our interpretation is that the irreversible component of the
dynamics acts to push the system away from regions where the free energy gradient is very steep, so the differences between very
fast and very slow processes tend to be smoothed out by the irreversibility.

\subsection{Accelerating macroscopic processes}
\label{sec:macro-theory}

In this section we consider hydrodynamic limits of interacting particle systems, as described by the macroscopic fluctuation
theory (MFT)~\cite{Bertini2015a}.  We will demonstrate that the large deviation result~\eqref{equ:i2-faster} has a counterpart at
the hydrodynamic level.  We also explore the geometrical interpretation of this result, and we connect our result to earlier work
related to SDEs that describe the motion of single particles~\cite{Rey-Bellet2015a}.

\subsubsection{Macroscopic Fluctuation Theory}
\label{sec:Macr-Fluct-Theory}

We first recall the core parts of the Macroscopic Fluctuation Theory (MFT). For a detailed review we refer
to~\cite{Bertini2015a}. Let $\Lambda \subseteq \mathbb R^d$ be a connected domain with boundary $\partial \Lambda$. {
  For simplicity, we choose here the domain $\Lambda=[0,1]^d$.} If we consider a microscopic particle process {  (indexed by $L$), its description
within MFT involves two random fields, the empirical particle density$\rho_t^L$ and the empirical current $j_t^L$.  Roughly speaking,
for $x\in\Lambda$ then $\rho_t^L$ is the local particle density and $j_t^L$ is a vector that indicates the rate of
particle flow.  

The idea of the hydrodynamic limit is that if we observe an interacting particle system on suitable large scales of length and time, then the
system can be described in terms of sufficiently smooth fields $\rho$ and $j$, instead of requiring a microscopic description in
which all particle positions are taken into account. The deterministic quantities $\rho$ and $j$ are then related by a continuity equation given by
\begin{equation}
  \partial_t \rho_t + \nabla \cdot j_t = 0 .
\end{equation}}

The domain $\Lambda$ is fixed in the hydrodynamic limit.  The relevance to large length and time scales in the microscopic model
is that one considers a large number of particles $N$ within a domain $\Lambda_L$ of linear size $L$.  One takes $N,L$ to infinity
together for a fixed density $\tilde\rho_0=N/L^d$.  The domain $\Lambda$ is obtained by rescaling the (increasingly large) domain
$\Lambda_L$, so that $\Lambda$ remains fixed as $L\to\infty$.

\renewcommand{\j}{{\color{red}YUK}}  %% RLJ: I added this to suppress the funny straight js, use \jmath instead

Within this hydrodynamic limit, the behaviour of the system on suitable (large) scales of space and time becomes increasingly
deterministic.  For example, given a time interval $[0,T]$ and initial and final densities $\rho_0$ and $\rho_T$, the probability
measure for paths connecting these initial and final states concentrates (in the hydrodynamic limit) on a single most likely path.
This result can be expressed as a large deviation principle for paths, which can, following \cite{Bertini2015a}, be written as
{
\begin{equation}
  \label{equ:ldp-path}
  \mathrm{Prob}\left[ (\rho_t^L,j_t^L)_{t\in[0,T]} \approx (\rho_t,j_t)_{t\in[0,T]} \right] \asymp {\rm e}^{-L^d{\cal I}
    (\rho,j)}
\end{equation}}
with
\begin{equation}
  \label{eqn:LDP1}
  \mathcal I(\rho, j)
  =\frac 14\int_{0}^{T}\int_\Lambda (j_t-J(\rho_t))\cdot\chi(\rho_t)^{-1}(j_t-J(\rho_t))\;\!dx\;\!dt
\end{equation}

whenever $\partial_t \rho_t = -\nabla\cdot j_t$ is satisfied, and $\mathcal I(\rho, j) =\infty$ otherwise. { We refer
  the reader to the review~\cite{Bertini2015a} for details on the validity of~\eqref{equ:ldp-path} for a large class of particle
  systems including the symmetric exclusion process and zero-range processes~\cite{Kipnis1999a,Spitzer1970a}.}

Note that in contrast to the large deviation principle in Sec.~\ref{sec:largedev-intro} which is concerned with large times, this principle
involves a limit of large $L$, with a fixed time interval $[0,T]$. 

Physically, we interpret $J(\rho_t)$ in~\eqref{eqn:LDP1} as
the most likely current field $j_t$, given that the system has density $\rho_t$.  Within MFT, the current is
assumed~\cite[Eq.~(2.6)]{Bertini2015a} to have the form
\begin{equation}
  \label{equ:JJ}
  J(\rho) =  -D(\rho) \nabla\rho + \chi(\rho) E , 
\end{equation}
where $\chi(\rho)$ and $D(\rho)$ are symmetric positive definite $d\times d$ matrices that depend on the local density $\rho$, and
$E$ is a fixed ($x$-dependent) vector field.

Physically, $D$ and $\chi$ correspond to a density-dependent diffusivity and mobility, while $E$ corresponds to an external force.
For a given interacting particle system, the parameters $D$, $\chi$ and $E$ can (in principle) be derived from the microscopic
rules of the model.  These parameters (along with appropriate boundary conditions associated with $\partial\Lambda$) fully specify
the rate function~\eqref{eqn:LDP1} and they fully describe the hydrodynamic limit of the interacting particle system. To fix the
ideas precisely, it may be useful to note that $J(\rho)$ in~\eqref{equ:JJ} is itself a field, whose value at position
$x\in\Lambda$ is $J(\rho)(x) = -D(\rho(x)) \nabla\rho(x) + \chi(\rho(x)) E(x)$.

Since $J(\rho)$ is the most likely current for a given density $\rho$, it follows that for a given initial condition, the path
measure is dominated by paths $(\rho_t)_{t\in[0,T]}$ which solve $\partial_t \rho = - \nabla \cdot J(\rho)$.  These paths have
${\cal I}=0$ and are said to satisfy the hydrodynamics.

As well as the large-deviation principle for paths~\eqref{equ:ldp-path}, the MFT also provides a large-deviation principle for the
fluctuations of the instantaneous density, in the steady state of the system.  That is, if the time $T$ is large enough that the
system has converged to its steady state, one has
{
\begin{equation}
  \label{equ:quasi} 
  \mathrm{Prob}[\rho_T^L \approx \rho] \asymp {\rm e}^{-L^d{\cal V}(\rho)}, 
\end{equation}}
where $\cal V$ is called the quasipotential: it determines the probability of fluctuations in the density.  {
  Eq.~\eqref{equ:quasi} is derived under the assumption that the adjoint dynamics satisfy a further Large Deviation principle for
  a rate functional $\mathcal I^\ast$. We refer to chapter II in \cite{Bertini2015a} for a detailed discussion.}

We assume throughout that our system has a unique steady state, for which the most likely ($x$-dependent) density is $\rhobar$.  In this case
${\cal V}(\rhobar)=0$ and ${\cal V}(\rho)>0$ for all $\rho\neq\rhobar$.

\subsubsection{Reversible and irreversible systems}
\label{sec:Reversible-irrev-sys}

For the microscopic dynamics, we already observed that the detailed balance condition~\eqref{eqn:DB} describes an important
special case.  By starting from this case, the generator was decomposed into two components~\eqref{equ:LSA}, corresponding to a
reversible process and a correction term that captures the irreversibility.  At the hydrodynamic level, there is a corresponding
decomposition which takes place at the level of the current: one writes
\begin{equation} 
  \label{equ:JSA}
  J=J_S + J_A . 
\end{equation}
The symmetric part of the current is defined~\cite[Equ.~(2.19)]{Bertini2015a} as
\begin{equation}
  \label{equ:JS}
  J_S(\rho) = - \chi(\rho) \nabla \frac{\delta {\cal V}}{\delta \rho},
\end{equation}
where $\frac{\delta {\cal V}}{\delta \rho}$ denotes the functional derivative of the quasipotential introduced in
Eq.~\eqref{equ:quasi}.  The antisymmetric part of the current is orthogonal to $J_S$, in the sense that
\begin{equation}
  \label{equ:HJ} 
  \int_\Lambda J_A(\rho) \cdot \chi^{-1}(\rho) J_S(\rho) \;\! dx = 0 , 
\end{equation}
which is sometimes referred to as a \emph{Hamilton-Jacobi equation}.  Note that this is an orthogonality in the space of fields:
the presence of the integral implies that the currents $J_S$ and $J_A$ do not have to be orthogonal at any specific point $x$.  We
note that on combining~\eqref{equ:JS} and~\eqref{equ:HJ}, one has
$\int_\Lambda J_A(\rho) \cdot \nabla \frac{\delta {\cal V}}{\delta \rho} \;\! dx = 0 ;$ integrating by parts and
using~\eqref{equ:JSA} one sees that
\begin{equation}
  \partial_t {\cal V} = \langle \partial_t \rho , \frac{\delta {\cal V}}{\delta \rho} \rangle 
  = \langle {\rm div} J,  -\frac{\delta {\cal V}}{\delta \rho} \rangle = -\langle J_S,\chi^{-1} J_S\rangle
\end{equation} 
which is independent of $J_A$.  Hence the quasipotential is non-increasing for paths satisfying the hydrodynamics, and (for any
given $\rho_t$) its time derivative is independent of $J_A$.

The special case in which the microscopic model is reversible has two implications for the hydrodynamic limit as described by MFT.
First, reversible models lead to $J_A=0$, so $J=J_S$.  Second, assuming that correlations in the particle model occur only on the
microscopic scale, the quasipotential within the MFT takes the simple (local) form~\cite[Equ.~(2.25)]{Bertini2015a}
\begin{equation}
  \label{equ:quasi-eq}
  {\cal V}(\rho)=\int_\Lambda \Bigl[ f(\rho)-f(\rhobar)-f'(\rhobar)(\rho-\rhobar) \Bigr]\;\! dx,
\end{equation} 
where $f(\rho)$ is the free energy per unit volume.  (The dependence of $f$ on $\rho$ is fixed by the microscopic model of
interest; note also that both $\rho$ and $\rhobar$ depend in general on the position $x$, but $f$ is a local function
$f(\rho)(x)=f(\rho(x))$.)

Hence for reversible microscopic models, the hydrodynamic current obeys
\begin{equation}
  J(\rho) = J_S(\rho) = -\chi(\rho) f''(\rho) \nabla \rho + \chi(\rho)  \nabla f'(\rhobar) .
\end{equation}
In this case consistency with~\eqref{equ:JJ} requires
\begin{equation} 
  \label{equ:einstein}
  E=\nabla f'(\rhobar), \qquad D(\rho)=f''(\rho)\chi(\rho) . 
\end{equation}
The second of these conditions is required within MFT. It is known as the local Einstein relation since it relates the mobility
$\chi$ to the diffusion constant $D$.  { Note that the equations~\eqref{equ:einstein} are consistent with the
  hydrodynamic limit for a large class of particle systems of `gradient type', see~\cite[Chapter VIII, Section G]{Bertini2015a}.}

We end this section with a brief comment on the boundary conditions within MFT. If the boundary is associated with coupling of the
system to a reservoir at chemical potential $\lambda$, the density at the boundary is fixed such that $f'(\rho)=\lambda$.  If
particles cannot penetrate the boundaries, one requires $D\nabla\rho=\chi E$ (and $j=0$) on $\partial\Lambda$.  Paths (or
configurations) that do not respect these boundary conditions have ${\cal I}=\infty$.

\subsubsection{Breaking detailed balance accelerates convergence}
\label{sec:hydro-accel}

We now state the sense in which breaking detailed balance accelerates convergence of interacting particle systems at the
hydrodynamic scale.  For the microscopic models, we compared two Markov chains, with the same invariant measure and generators
$\LL$ and $\LL_S$.  At the hydrodynamic scale, we will compare two systems with the same quasipotential (this corresponds to
comparing two microscopic models with the same invariant measure).  One system is irreversible and has a general $J$ given
by~\eqref{equ:JJ}; the second system is reversible and so $J_A=0$.  {In order to ensure a fair comparison, we also assume that the two models have the same mobility $\chi(\rho)$: for Markov processes the equivalent condition was that we always compared models with the same ${\cal L}_S$.}  Since $\cal V$ {and $\chi$} are the same for both models, they both have the
same symmetric current $J_S$ which is given by~\eqref{equ:JS}.

For each of these systems, we consider the large deviations of the time-averaged density, following Sec.~\ref{sec:largedev-intro}.
Large deviation principles of the form
{
\begin{equation} 
  \label{equ:ldp-mft-l2}
  \mathrm{Prob}\left[\frac1T \int_0^T \rho^L_t(\cdot) \;\!dt  \approx \rho(\cdot) \right] \asymp {\rm e}^{-T L^d I_2(\rho)}
\end{equation}}
apply in both reversible and irreversible models. This large deviation principle applies on taking the large-$T$ limit after the
hydrodynamic limit: one should take $L\to\infty$ before $T\to\infty$.  
%We now show how breaking detailed balance can accelerate
%convergence to equilibrium on the hydrodynamic level, in the sense of large deviation principles. 
To {obtain bounds on $I_2$}, we introduce the
so-called level-2.5 large-deviation principle for the joint fluctuations of the empirical current and empirical measure \cite{Barato2015,bertini2015b}.  That is,
{
\begin{equation}
  \label{equ:ldp-mft-l2.5}
  \mathrm{Prob}\left[\frac1T \int_0^T \rho^L_t(\cdot)\;\!dt \approx \rho(\cdot), \frac1T \int_0^T j^L_t(\cdot) \;\!dt \approx  j(\cdot)\right] 
  \approx {\rm e}^{-T L^d I_{2.5}(\rho,j)} . 
\end{equation}}

If we assume that the paths that dominate the level-2.5 LDP are constant in time, the relevant rate function can be obtained
from~\eqref{equ:ldp-path} as
  \begin{equation}
    \label{eqn:level2.5}
    I_{2.5}(\rho,j) = \frac{1}{4} \int_\Lambda (j-J(\rho))\cdot\chi(\rho)^{-1}(j-J(\rho))\;\!dx
  \end{equation}
  if $\nabla\cdot j=0$, and $I_{2.5}=\infty$ otherwise.  The assumption of time-independent paths is equivalent to assuming that
  no dynamical phase transition takes place~\cite{Bertini2001a,Bodineau2004a}.  Using this assumption, we now calculate a bound
  (Prop.~\ref{prop:ldpinequal}) for the level-2 rate functionals, which is analogous to~\eqref{equ:i2-faster} in the microscopic
  case.

\begin{proposition}
  \label{prop:ldpinequal} 
  Let the level-2.5 rate functional be given by~\eqref{eqn:level2.5} and let $I_2$ be the level-2 large deviation rate functional
  obtained from $I_{2.5}$ by contraction. We write $I_2^{\rm rev}$ for this rate functional if the current is symmetric,
  $J = J_S$, and we write $I_2^{\rm irrev}$ for the rate functional for the general case $J = J_S + J_A$ as
  in~\eqref{eqn:currentdecomposition}. Then
  \begin{equation}
    I_2^{\rm irrev}(\rho) \ge I_2^{\rm rev}(\rho). 
    \label{equ:i2-faster-mft}
\end{equation}
\end{proposition} { {\it Remark:} Note that this result will be strengthened later. We will obtain in
  equation~\eqref{eqn:LDPfunctional2} an exact identity for $I_2^{\rm irrev}$ as the sum of $I_2^{\rm rev}$ and a non-negative
  quantity.

\begin{proof}
We write $I_2$ for $I_2^{\rm irrev}$.
The rate functional at level-$2$ can be obtained by a contraction of the level-2.5 rate functional,
\begin{equation}\label{eqn:2.5to2ldp}
  I_2(\rho)=\inf_{j:\nabla \cdot j=0}I_{2.5}(\rho,j) . 
\end{equation}
Note that $I_{2.5}(\rho,j)$ as given in equation \eqref{eqn:level2.5} is (using \eqref{equ:HJ}) equal to the sum of the following three summands:
\begin{multline}\label{eqn:1st_split_2.5}
  %\begin{split}
    \frac{1}{4} \int_\Lambda (j-J_S(\rho))\cdot\chi(\rho)^{-1}(j-J_S(\rho))\;\!dx\\
    +\frac{1}{4} \int_\Lambda (j-J_A(\rho))\cdot\chi(\rho)^{-1}(j-J_A(\rho))\;\!dx
    -\frac{1}{4} \int_\Lambda j\cdot\chi(\rho)^{-1}j\;\!dx.
  %\end{split}
\end{multline}
The summand in the first line coincides with the symmetric rate functional $I_{2.5}^{\rm rev}(\rho,j)$ and the second line is the part that corresponds to the anti-symmetric dynamics. Dropping the first summand in the second line (which is non-negative), we obtain 
\begin{equation}\label{eqn:2.5ldp_ineq}
I_{2.5}(\rho,j)\ge 
\frac{1}{4} \int_\Lambda (j-J_S(\rho))\cdot\chi(\rho)^{-1}(j-J_S(\rho))\;\!dx
-\frac{1}{4} \int_\Lambda j\cdot\chi(\rho)^{-1}j\;\!dx.
\end{equation}
An expansion of the square shows that the right hand side is equal to
\[
\frac{1}{4} \int_\Lambda J_S(\rho)\cdot\chi(\rho)^{-1}J_S(\rho)\;\!dx
-\frac 12 \int_\Lambda J_S(\rho)\cdot\chi(\rho)^{-1} j\;\!dx,
\]
and the last summand vanishes under the assumption that $\nabla\cdot j=0$, as by equation \eqref{equ:JS}
\begin{equation}\label{eqn:orth_j_J_S}
\int_\Lambda J_S(\rho)\cdot\chi(\rho)^{-1} j\;\!dx
=- \int_\Lambda 
\nabla \frac{\delta {\cal V}}{\delta \rho}\cdot j\;\!dx= \int_\Lambda  \frac{\delta {\cal V}}{\delta \rho}\nabla\cdot j\;\!dx =0.
\end{equation}
We obtain with \eqref{eqn:2.5to2ldp} that 
\[
  I_2(\rho)=\inf_{j:\nabla \cdot j=0}I_{2.5}(\rho,j)\ge \frac{1}{4} \int_\Lambda J_S(\rho)\cdot\chi(\rho)^{-1}J_S(\rho)\;\!dx.
\]
{To establish (\ref{eqn:2.5to2ldp}) we now show that the right hand side of this expression coincides with $I_2^{\rm rev}(\rho)$.}
%We are finished once the right hand side is identified with $I_2^{\rm rev}(\rho)$.
Note that again for $j$ such that $\nabla\cdot j=0$, by the same argument as in \eqref{eqn:orth_j_J_S}, the reversible level-2.5 rate functional is equal to
\begin{equation}\label{eqn:ldp_2.5_sym_xx}
I_{2.5}^{\rm rev}(\rho,j)=
\frac 14 \int_\Lambda j\cdot \chi(\rho)^{-1} j \;\! dx + \frac 14 \int_\Lambda J_S(\rho)\cdot \chi(\rho)^{-1}J_S(\rho)\;\! dx.
\end{equation}
As one would expect for the reversible case, the infimum in \eqref{eqn:2.5to2ldp} is clearly attained for a vanishing current ($j=0$), so that
\begin{equation}\label{eqn:ldp_2_sym_xx}
I_2^{\rm rev}(\rho)= \frac{1}{4} \int_\Lambda J_S(\rho)\cdot\chi(\rho)^{-1}J_S(\rho)\;\!dx,
\end{equation}
which completes the proof.
\end{proof}
}

Of course, given the acceleration at the microscopic scale, the result~\eqref{equ:i2-faster-mft} that this acceleration is
preserved at the hydrodynamic limit may not be surprising.  However, we show below that the geometric structure
underlying the MFT allows some stronger results for this acceleration to be established.

\subsubsection{Splitting the current}
\label{sec:Splitting-current}

To understand the geometrical origin of~\eqref{equ:i2-faster-mft} in more detail, we now show that as well as the
decomposition~\eqref{equ:JSA}, the antisymmetric current $J_A$ has a further decomposition into two parts which are orthogonal to
each other, and are both orthogonal to $J_S$. (Here, orthogonality should be understood in the sense of~\eqref{equ:HJ}.)

{  We consider the problem
\begin{equation} 
  \label{equ:supPsi}
  \nabla\cdot\bigl(\chi(\rho) \nabla\psi \bigr) = -\nabla\cdot J_A(\rho),
\end{equation}
with the boundary condition $\psi = 0$ on $\partial\Lambda$.
For any fixed $\rho$ (such that $\chi(\rho)$ and $J_A(\rho)$ are sufficiently regular) equation \eqref{equ:supPsi} has a unique strong solution $\psi$ (see for example Theorem~6.24 in \cite{GilbargTrudinger2001}).} This solution $\psi$ is therefore a functional of $\rho$ we will denote with $\psi(\rho)$.  Eq.~\eqref{equ:supPsi}
motivates us to decompose $J_A(\rho)$ as
\begin{equation} 
  \label{equ:JAdecomp}
  J_A(\rho) = -\chi(\rho) \nabla\psi(\rho)  + J_F(\rho) , 
\end{equation}
where $J_F(\rho)$ is a new vector field, which is again a functional of $\rho$.  From~\eqref{equ:supPsi} we see that
\begin{equation}
  \label{eqn:divergencefree}
  \nabla \cdot J_F(\rho) = 0
\end{equation}
for all $\rho$. 

We arrive at the following structure for the hydrodynamic current:
\begin{equation}
  \label{eqn:currentdecomposition}
  J(\rho) = J_S(\rho) -\chi(\rho) \nabla \psi(\rho) + J_F(\rho).
\end{equation}
Of the three terms on the right hand side, the first is familiar as the symmetric current, while the third is divergence free and
so does not transport any density.  The remaining term (involving $\psi$) specifies how the density is transported by the
antisymmetric current, and also determines the large deviations at level-$2$. The latter will be established below as a consequence of the following proposition.

\begin{proposition}
  \label{prop:splittingofcurrent}
  The three terms on the right hand side of Eq.~\eqref{eqn:currentdecomposition} are all orthogonal in the sense of   Eq.~\eqref{equ:HJ}. Moreover, $J_S(\rho)$ and $-\chi(\rho)\nabla\psi(\rho)$ are orthogonal to all divergence free vector fields that vanish on the boundary.
\end{proposition}
\begin{proof}
  Consider first the orthogonality between $J_F(\rho)$ and $\chi(\rho)\nabla\psi(\rho)$.  One has $\psi(\rho)|_{\partial\Lambda}=0$ so integration by
  parts yields
\begin{equation*} 
  \int_\Lambda \chi(\rho)\nabla\psi(\rho) \cdot \chi^{-1}(\rho) J_F(\rho) \;\! dx = -\int_\Lambda \psi(\rho) \nabla \cdot J_F(\rho) \;\! dx = 0 
\end{equation*}
where the second equality follows from~\eqref{eqn:divergencefree}.  Hence $J_F(\rho)$ and $\chi(\rho)\nabla\psi(\rho)$ are orthogonal in the sense
of~\eqref{equ:HJ}.

Following the same method but replacing $\psi$ by $\delta {\cal V}/\delta \rho$ shows that $J_F(\rho)$ is orthogonal to
$J_S(\rho)=-\chi(\rho)\nabla(\delta {\cal V}/\delta \rho)$, where we used $(\delta {\cal V}/\delta \rho)|_{\partial\Lambda}=0$, as discussed
in~\cite{Bertini2015a}.

Finally, using the orthogonality relation~\eqref{equ:HJ} and $J_A(\rho)=-\chi(\rho) \nabla \psi(\rho) + J_F(\rho)$ yields
\begin{equation*}
  \int_\Lambda \chi(\rho)\nabla\psi(\rho) \cdot \chi^{-1}(\rho) J_S \;\! dx  = \int_\Lambda J_F(\rho) \cdot \chi^{-1}(\rho) J_S(\rho) \;\! dx.
\end{equation*} 
The right hand side vanishes by orthogonality of $J_S(\rho)$ and $J_F(\rho)$, so $\chi(\rho) \nabla \psi(\rho)$ is orthogonal to $J_S(\rho)$, as required.
\end{proof}

Combining Eq.~\eqref{eqn:currentdecomposition} and Eq.~\eqref{eqn:divergencefree}, the dynamics of the density is given by
\begin{equation}
  \label{eqn:gradientdynamics}
  \partial_t \rho = \nabla\cdot \bigl(\chi(\rho) \bigl[\nabla  \tfrac{\delta \mathcal V}{\delta \rho} 
  + \nabla \psi(\rho)\bigr]\bigr)
\end{equation}
The first term on the right hand side describes steepest descent (gradient flow) of the quasipotential, within a (modified)
Wasserstein metric~\cite{Adams2013a,Jack2014a}.  The second term describes a current that is orthogonal to the gradient flow
(within the same metric), and leads to an evolution of $\rho$ within the level sets of the quasipotential: this is the geometric
result anticipated in Sec.~\ref{sec:markov-discuss}, but in this hydrodynamic setting the geometrical objects are more explicit.

{
We now derive exact formulas for the level-2.5 and level-2 rate functionals based on the splitting in Proposition~\ref{prop:splittingofcurrent}.
\begin{proposition}
Let the level-2.5 large deviation rate functional be given by~\eqref{eqn:level2.5}. Further let $\rho$ be such that equation~\eqref{equ:supPsi} has a unique classic solution (up to a constant) and $j$ such that $\nabla\cdot j= 0$. Then,
\begin{equation}
\begin{split}\label{eqn:prop_ldp_2.5}
I_{2.5}(\rho,j) = &\frac 14 \int_\Lambda (j-J_F(\rho))\cdot\chi(\rho)^{-1}(j-J_F(\rho))\;\!dx\\
&+ \frac 14
  \int_\Lambda  \nabla   \tfrac{\delta \mathcal V}{\delta \rho}\cdot \chi(\rho) \nabla   
  \tfrac{\delta \mathcal V}{\delta \rho}\;\!dx
  + \frac 14\int_\Lambda  \nabla \psi(\rho)\cdot \chi(\rho) \nabla \psi(\rho)\;\!dx.
\end{split}
\end{equation}
Moreover, the level-2 rate functional is given by
\begin{equation}\label{eqn:LDPfunctional2}
I_2(\rho) = \frac 14
  \int_\Lambda  \nabla   \tfrac{\delta \mathcal V}{\delta \rho}\cdot \chi(\rho) \nabla   
  \tfrac{\delta \mathcal V}{\delta \rho}\;\!dx
  + \frac 14\int_\Lambda  \nabla \psi(\rho)\cdot \chi(\rho) \nabla\psi(\rho)\;\!dx.
\end{equation}
\end{proposition}

\begin{proof}
The proof of equation~\eqref{eqn:prop_ldp_2.5} follows from Proposition~\ref{prop:splittingofcurrent} and the representation of the rate functional \eqref{eqn:1st_split_2.5}. The second result~\eqref{eqn:LDPfunctional2} follows readily as $j=J_F(\rho)$ is the minimiser of~\eqref{eqn:prop_ldp_2.5}.
\end{proof}
Note that these results are consistent with~\eqref{eqn:ldp_2.5_sym_xx} and~\eqref{eqn:ldp_2_sym_xx}, where the minimising current was given by $j=0$. In the general case, the minimising current is given by $j=J_F(\rho)$. 

We moreover can
 recognise the first term on the right hand side of \eqref{eqn:LDPfunctional2} }as $I_2^{\rm rev}(\rho)$, so the second term on the right hand side is an
exact formula for the difference in rate for reversible and irreversible processes.  This shows that the convergence rate for the
irreversible process is strictly faster, unless the force $(-\nabla\psi)$ vanishes.
 We recognise this as a condition that the antisymmetric part of the current contributes to the time derivative of the density (otherwise the convergence to equilibrium of the density can not be accelerated).

Note that the objects $\nabla \tfrac{\delta \mathcal V}{\delta \rho} $ and $\nabla\psi$ should be interpreted as forces acting
in the space of densities.  In order to sustain a large deviation of the density, the stochastic forces within the system must act
to resist these (deterministic) forces.  One sees from~\eqref{eqn:LDPfunctional2} that the probability of this rare event (or
large deviation) is given by the norms of the two forces, within a metric that depends on the mobility $\chi$.

\subsubsection{An example}
\label{sec:example}

We have discussed the status of the MFT as a theory for the hydrodynamic limit of interacting particle systems.  For a concrete
example of this approach, we consider an interacting particle model known as the zero-range process (ZRP)~\cite{Spitzer1970a}.  A
microscopic description of the ZRP is given in Sec.~\ref{sec:zero-range-process}.  For the purposes of this section, the important
features of the ZRP are that its hydrodynamic limit is described by the MFT and that irreversible ZRPs have local quasipotentials
of the form~\eqref{equ:quasi-eq}.  This latter fact allows straightforward comparison between reversible and irreversible models
with the same quasipotential.

The hydrodynamic limit of the ZRP is a non-linear drift-diffusion
\begin{equation}
  \label{eqn:hydro}
  \partial_t \rho  = \Delta\phi(\rho) -  \nabla\cdot\bigl(\phi(\rho)E\bigr),
\end{equation}
where $\phi$ is a function that depends on the local density [that is, $\phi(\rho)(x)=\phi(\rho(x))$], and $E$ is a drift term.
The specific function $\phi$ that appears in the MFT depends on how the particles interact within the ZRP. A formal derivation of
this hydrodynamic limit can e.g. be found in \cite{Bertini2015a}.  If $\phi(\rho)=\rho$, then the model
corresponds to drift-diffusion of non-interacting particles.

One sees immediately from~\eqref{eqn:hydro} that the hydrodynamic current is given by~\eqref{equ:JJ} with $\chi(\rho)=\phi(\rho)I$
and $D(\rho)=\phi'(\rho)I$, where $I$ is the identity matrix.  Moreover, the quasipotential for the ZRP is given
by~\eqref{equ:quasi-eq} with $f'(\rho)=\log\phi(\rho)$, consistent with~\eqref{equ:einstein}.  The ZRP may be either reversible or
irreversible: one sees that reversible ZRPs lead to $E=-\nabla V$ for some potential $V$. In this case~\eqref{equ:einstein} shows
that $V(x)=\log(\phi(\rhobar(x)))+\lambda$, where $\rhobar$ is the steady state density profile and $\lambda$ is a constant
(independent of $x$).  Hence one identifies the irreversible current as
$J_A(\rho)=J(\rho)-J_S(\rho)=\phi(\rho) \left[ E + \nabla \log \phi(\rhobar)\right]$.

Examining the rate function~\eqref{eqn:LDPfunctional2} for the specific case of the ZRP, one can interpret the result as a
generalisation of a result in~\cite{Rey-Bellet2015a}.  One has $\delta {\cal V}/\delta \rho = \log\phi(\rho) - \log\phi(\rhobar)$.
Hence
\begin{equation}
  I_2(\rho) =\int \Bigl(\bigl|\nabla \log\bigl(\tfrac{\phi(\rho)}{\phi(\bar\rho)}\bigr)\bigr|^2 
  +|\nabla\psi(\rho)|^2\Bigr)\phi(\rho)\;\! dx,
\end{equation}
where $\psi$ is the solution of
$\nabla\cdot\bigl(\phi(\rho) \nabla\psi \bigr) = -\nabla \cdot[ \phi(\rho) ( E+\nabla\log\phi(\rhobar))]$.  If we now consider the
special case $\phi(\rho)=\rho$ then we recover the same rate function as in Theorem 2.2 of Ref.~\cite{Rey-Bellet2015a}: the
non-gradient force $C$ in that work is here replaced by $E+\nabla\log\rhobar$ (note that this is independent of $\rho$).  The
condition that $\nabla\cdot(\rhobar C)=0$ -- which ensures that the invariant measure is unchanged by breaking detailed balance --
is satisfied within the MFT because $\nabla\cdot J_A(\rhobar)=0$ and setting $\phi(\rho)=\rho$ yields
$J_A(\rhobar)=\rhobar(E+\nabla\log\rhobar)$.

Note however the setting discussed in this work is different to that in~\cite{Rey-Bellet2015a}: here we consider the hydrodynamic
limit of many particles on a lattice while that work considers a single particle in a compact manifold without boundary.  For
non-interacting particles, the result is the same: the reason that for the many-particle system, the rate function $I^N$
associated with all the particles undergoing the same rare fluctuation is equal to $NI^{1}$.  So the only difference between the
one-particle and many-particle systems arises in the prefactors (speeds) of the large deviation
principles~\eqref{equ:ldp-mft-l2},~\eqref{equ:ldp-mft-l2.5}.

\section{Application to the zero-range process, and numerical results}
\label{sec:numerics}

\subsection{The zero-range process}
\label{sec:zero-range-process}

The ZRP~\cite{Spitzer1970a} is a system in which interacting particles move on a finite lattice
$\Lambda_L=\{0,\dots,L-1\}^d \subseteq \mathbb Z^d$ where $L\in\mathbb N$ is the linear system size.  The particles are assumed to
be indistinguishable and each particle is located at one of the sites $x\in\Lambda_L$.  The number of particles on site $x$ is
$\eta(x)$ and the configurations of the system are $\eta=(\eta(x))_{x\in\Lambda_L}$.  We will assume that the total number of
particles is conserved such that no particles are added or removed over time.

The interaction of the particles is encoded in a function $g(k)$, with $g(0)=0$.  The rate of particle transfer from site $x$ to site $y$ is $g(\eta(x))c(x\!\to\!y)$, where the function $c$
determines the connectivity of the sites. The case $g(k)=k$ corresponds to non-interacting particles.  The model is referred to as
zero-range because particles interact only when they are on the same site. For example, if $g(k)=k^\alpha$ for $k > 0$, then
$\alpha<1$ means particles on the same site attract each other (suppressing jumps away from that site) while $\alpha>1$ means that
particles on the same site tend to repel each other.

\subsubsection{Reversible and irreversible ZRP}
\label{sec:Reversible-irrev-ZRP}

The behaviour of the ZRP depends strongly on the choice of the connectivity function $c$ as well as the interaction function $g$.
We assume that particles hop only to nearest neighbour sites, so $c(x\to y)>0$ only if $x$ and $y$ are nearest neighbours.  At the
boundaries of the lattice, the system has either reflecting boundaries (particles cannot leave the lattice) or periodic
boundaries.

It is easily verified that the model obeys the detailed balance condition~\eqref{eqn:DB} if one takes (for nearest neighbour
sites)
\begin{equation}
  \label{equ:cxy-rev}
  c(x \to y) = {\rm e}^{\frac12[V(x)-V(y)]}
\end{equation}
for some potential function $V$.  In this case the model is reversible.

To arrive at a class of irreversible models, we take
\begin{equation} 
  \label{equ:cxy-irrev}
  c(x \to y) = {\rm e}^{\frac12[V(x)-V(y)]} + k_{x,y} {\rm e}^{V(x)}
\end{equation} 
with $k_{x,y} =-k_{y,x} $.  In this case positivity of transition rates requires $|k_{x,y}|<{\rm e}^{-\frac12[V(x)+V(y)]}$ for all
$x,y$.  We show below that taking $k\neq 0$ corresponds to breaking of detailed balance, in the sense of~\eqref{equ:LSA}.

\subsubsection{Generator and invariant measure}
\label{sec:Gener-invar-meas}

We denote the configuration of the ZRP at time $t$ with $\eta_t$. The generator acts on the test function $f$ as
\begin{equation}
  \label{eqn:DBdynamics}
  \mathcal L f(\eta) = \sum_{x,y\in\Lambda_L} (f(\eta^{x,y})-f(\eta)) g(\eta(x))c(x\!\to\!y).
\end{equation}
Here $\eta^{x,y}$ denotes the configuration obtained from $\eta$ by removing one particle from position $x$ and adding it at
position $y$. If $\eta(x)=0$ we simply set $\eta^{x,y}=\eta$ and hence leave the configuration unchanged.

Note that the ZRP as defined so far is reducible, since the number of particles is a conserved quantity under the dynamics.  This
setting is useful because it is easily verified (directly from the definition~\eqref{eqn:DBdynamics} and using that the invariant
measure $\pi$ satisfies $\sum_\eta \pi(\eta) \LL f(\eta)=0$ for all $f$) that the reversible model with rates defined
in~\eqref{equ:cxy-rev} has a family of invariant measures, { the so called grand-canonical measures, which are parameterised by the chemical potential $\lambda$ and given by
\begin{equation} 
  \label{equ:pi-gc}
  \pi_{\mathrm{grand}}^\varphi(\eta) = \prod_{x\in\Lambda_L} \frac{\varphi(x)^{\eta(x)}}{z(\varphi(x)) g!(\eta(x))}
\end{equation} }
with the fugacity $\varphi(x)=\mathrm e^{-V(x)-\lambda}$ for some $\lambda\in\mathbb R$; the notation $g!(k)$ indicates the generalised factorial
$g!(k):=\prod_{i=1}^k g(i)$ [with $g!(0)=1$] and $z(\varphi)=\sum_{k=0}^\infty \frac{\varphi^{k}}{g!(k)}$ is a normalisation
constant~\cite{Evans2005a,Kipnis1999a}. {We here assume that $V$, $\lambda$ and $g(\cdot)$ are such that $z(\varphi(x))<\infty$ for all $x\in \Lambda_L$. This is in particular the case for any $V$ and $\lambda$, when $g(\cdot)$ satisfies $g(k)\ge ck$ for some constant $c>0$ \cite{Kipnis1999a}.}

On restricting the model to a fixed number of particles $N$, the invariant measure $\pi$ (which is called the canonical measure) can be obtained by a conditioning
of~\eqref{equ:pi-gc}.  Note that~\eqref{equ:pi-gc} has the structure of a product measure. Also if $g(k)=k$ then one recovers the
case of non-interacting particles and the local marginals of~\eqref{equ:pi-gc} are Poisson distributions.

To make the comparison between reversible and irreversible models described in Sec.~\ref{sec:markov-theory}, we require an
irreversible model whose invariant measure is~\eqref{equ:pi-gc}.  Again using that $\sum_\eta \pi(\eta) \LL f(\eta)=0$ for all
$f$, we take $f=\eta(x)$ to be the number of particles on site $x$, from which we see that the irreversible
rates~\eqref{equ:cxy-irrev} are also consistent with the invariant measure~\eqref{equ:pi-gc} if we take
\begin{equation}
  \label{eqn:condition1}
  \sum_{y:y\sim x} (k_{x,y}-k_{y,x})=0~\textrm{ for all }x,
\end{equation}
where the notation $y\sim x$ indicates that sites $x$ and $y$ are nearest neighbours.  (If we imagine a system with just one
particle, this constraint states that the rate of hopping onto site $x$ is balanced by the rate of hopping away from that site.
For the ZRP, this same balance condition ensures that the invariant measure~\eqref{equ:pi-gc} is still valid even for many
interacting particles.)

Finally then, the conditions on the perturbations $k_{x,y}$ required for a meaningful comparison between reversible and
irreversible models can be summarised as:
\begin{equation}
  \label{eqn:condition2}
  \sum_{y:y\sim x} k_{x,y}=0,~\textrm{  }~k_{x,y}=-k_{y,x},~\textrm{ and }~|k_{x,y}|<  \mathrm e^{-(V(x)+V(y))/2}.
\end{equation}
The rates $k_{x,y}$ can be interpreted as elements of a matrix, which coincides (up to the factor $1/2$) with the vorticity
matrix $\Gamma$ introduced in~\cite{Bierkens2015a}.

{ 
In terms of the splitting \eqref{equ:LSA} the symmetric part of the dynamics is given by
  $c_s(x \to y) = {\rm e}^{\frac12[V(x)-V(y)]}$ and the anti-symmetric part by $c_a(x \to y) = k_{x,y} {\rm e}^{V(x)}$, such that the symmetric part (corresponding to $\mathcal L_S$) is independent of $k_{x,y}$.
}

\subsubsection{Hydrodynamic limit}
\label{sec:Hydrodynamic-limit}

The hydrodynamic limit of the ZRP is defined as follows.  For a ZRP on a lattice $\Lambda_L$ with $L^d$ sites, one takes
$N=\lfloor \rho_0 L^d\rfloor$ particles, where $\rho_0$ is the average density.  The lattice $\Lambda_L$ is mapped into the domain
$[0,1]^d$ by identifying each site $x\in\Lambda_L$ with a position $\tilde x \in \Lambda$ with $\Lambda = [0,1]^d$.  Hence the
site $x$ with integer co-ordinates $(i,j,\dots)$ has a position $\tilde x=(i/L,j/L,\dots)$.  Roughly speaking, the density
$\rho_t(\tilde x)$ in the MFT is equal to the typical number of particles on site $x$, and the normalisation of the density is
$\int_\Lambda \rho_t(\tilde x) \;\! d\tilde x=\rho_0$.  The hydrodynamic limit corresponds to a sequence of models in which
$L\to\infty$ at fixed $\rho_0$, so $N\to\infty$.

The hydrodynamic limit corresponds to observing a system on increasingly large length and time scales.  Note that since the number
of sites in $\Lambda_L$ is diverging (proportional to $L^d$) in the hydrodynamic limit, the diffusion constant for a single
particle (in $\Lambda$) vanishes as $L^{-2}$.  For this reason, when the lattice $\Lambda_L$ is mapped into the fixed domain
$\Lambda$, it is also convenient to scale the hop rates for all particles, by taking $c(x\to y)\to L^2 c(x\to y)$.  This ensures
that the diffusive behaviour characteristic of the hydrodynamic limit is observed, and the hydrodynamic limit is consistent with
MFT.

To fix the hop rates between sites in the ZRP, one fixes a smooth potential function $\tilde V \colon \Lambda\to \mathbb{R}$ on the
hydrodynamic scale, and one considers a sequence of ZRPs of increasing sizes $L$ with potential functions
$V(x)=\tilde V(\tilde x)$, where $\tilde x$ is the image in $\Lambda$ of the discrete site $x\in \Lambda_L$.  Similarly one fixes
a vector field $\tilde k \colon \Lambda\to\mathbb{R}^d$ and takes $k_{x,y}=\tilde{k}(\tilde x)\cdot(\tilde y - \tilde x)$ where
the dot indicates a scalar product in $\mathbb{R}^d$.

{The relation between the ZRP and the MFT is discussed in e.g.~\cite{Bertini2002}, \cite{Hirschberg2015} and in the review paper~\cite{Bertini2015a}.} In particular, for both reversible and irreversible
ZRPs one arrives at the situation described in Sec.~\ref{sec:example}.  The hydrodynamic limit~\eqref{eqn:hydro} depends on the
drift function $E\colon\Lambda\to\mathbb{R}^d$ which is given by $E(\tilde x)=-\nabla \tilde V(\tilde x)+\tilde k(\tilde x)$.

The MFT description of the ZRP also depends on a function $\phi$ which can be obtained as the solution of
\begin{equation}
  \rho = \sum_{k=1}^\infty \frac{k\;\!\phi(\rho)^{k}}{z(\phi(\rho)) g!(k)}  .
\end{equation}
We identify the right hand side of this equation as the mean local density associated with the measure~\eqref{equ:pi-gc}, at 
fugacity $\varphi=\phi(\rho)$.
 
{The quasipotential $\mathcal V$ for the ZRP is given by \cite{Bertini2015a},
\begin{equation}\label{eqn:quasipotential_zrp}
  \mathcal V(\rho) = \int_\Lambda \biggl[\rho(x) \log\biggl(\frac{\phi(\rho(x))}{\phi(\bar\rho(x))}\biggr) - \log \biggl(\frac{z(\phi(\rho(x)))}{z(\phi(\bar\rho(x)))}\biggr)\biggr] dx.
\end{equation}
}
\subsection{Simulation results}
\label{sec:Simulation-results}

We present numerical results for one-dimensional and two-dimensional systems, showing how breaking detailed balance (that is,
taking $k_{x,y}\neq0$ in~\eqref{equ:cxy-irrev}) accelerates convergence to equilibrium. The simulations are performed using the
Gillespie algorithm~\cite{Gillespie1977a}.  The results illustrate several aspects of the theoretical analysis in
Sec.~\ref{sec:theory}.  First, the results of that section do not rely on how detailed balance is broken: we show that there are
several possible choices and discuss their consequences.  Second, our numerical results show in what contexts we expect to see
significant acceleration of the dynamics on breaking detailed balance, and in what contexts we expect the acceleration to be mild.

In all cases, we show results that are scaled to be consistent with the hydrodynamic limit. That is, we map the lattice
$\Lambda_L$ into $[0,1]^d$ and we rescale the microscopic hop rates by a factor of $L^2$ so as to recover diffusive behaviour in
the hydrodynamic limit.

{%
In practical situations where the rate of convergence to equilibrium is important, a common situation is that the potential function $V$ is not convex, but includes several (or many) minima, separated by high barriers.  From a physical perspective, the temperature of our systems is a parameter that has been absorbed into the function $V$. In general, high barriers are linked with long (Arrhenius) time scales that are proportional to ${\rm e}^{\Delta V}$.  In order to understand whether breaking detailed balance can accelerate convergence in such non-convex problems, we consider cases where the function $V$ has two minima, with longest time scale in the system corresponding to motion between these minima.
}

%In both one- and two-dimensional cases we consider double-well potentials, which include a global minimum and a second (local)
%minimum.  The two minima are separated by a barrier and the slowest relaxation mode of the system is characterised by transfer of
%particles between the minima, which requires crossing of the barrier.  Such barrier-crossing events often result in slow
%convergence to equilibrium in physical systems, so it is useful to consider the possibility that breaking detailed balance can
%accelerate convergence in this case.

\subsubsection{Characterisation of convergence}
\label{sec:Char-conv}

{We perform numerical simulations starting from a fixed (deterministic) initial condition $\eta_0$.}
To analyse convergence to equilibrium, we perform numerical simulations of the ZRP, and we track the time-dependence of several
different quantities.  For any configuration $\eta$, the mean potential energy is
\begin{equation} 
  \langle \eta, V\rangle = \sum_{x\in\Lambda_L} \eta(x) V(x). 
\end{equation}
%{For a fixed initial measure $\mu_0$ equal to a deterministic initial configuration (i.e. $\mu_0$ equal to a Dirac mass for a fixed initial configuration with a fixed number of particles)}, 
We generate several trajectories (sample paths) $\eta_t$ of the ZRP and we
estimate the mean potential energy
\begin{equation} 
  \hat{V}(t)=\mathbb{E}_{\mu_0}(\langle \eta_t, V\rangle)
\end{equation} 
by taking the mean value of $\langle \eta_t, V\rangle$ over these trajectories. { For systems of non-interacting particles (where $\phi(\rho)=\rho$), we also
 estimate the macroscopic relative entropy as
\begin{equation}\label{eqn:relent_000}
  D(t) = \sum_{x\in\Lambda_L} \mathbb{E}_{\mu_0}(\eta_t(x)) \log \biggl(\frac{\mathbb{E}_{\mu_0}(\eta_t(x))}{\mathbb{E}_\pi(\eta(x))}\biggr), 
\end{equation}
which can be seen as an approximation to the quasipotential, which is for an independent random walk given by
\[
\mathcal V(\rho_t) = \int_\Lambda \biggl[\rho_t(x) \log\Bigl(\frac{\rho_t(x)}{\bar\rho(x)}\Bigr) + \rho_t(x) - \bar\rho(x) \biggr]\;\!dx
=\int_\Lambda \rho_t(x) \log\Bigl(\frac{\rho_t(x)}{\bar\rho(x)}\Bigr) \;\!dx,
\]
where we used the fact that $z(\varphi) = \mathrm e^{-\varphi}$ in \eqref{eqn:quasipotential_zrp} and the last identity follows from the fact that the density is conserved: $\int_\Lambda \rho_t(x) dx = \int_\Lambda \bar\rho(x) dx$.
}

For numerical purposes, we estimate $\mathbb{E}_{\mu_0}(\eta_t(x))$ as
the average occupancy of site $x$ over the sample paths that we generate, and we calculate $\mathbb{E}_\pi(\eta(x))$ by direct
construction of the invariant measure (whenever possible). Finally, we estimate the Gibbs entropy
\begin{equation}
  S(t) = -\sum_x \mathbb{E}_{\mu_0}(\eta_t(x)) \log \mathbb{E}_{\mu_0}(\eta_t(x)), 
\end{equation}
which is large if particles are delocalised throughout the system, and small if they are concentrated on a small number of sites.
Again, we estimate $\mathbb{E}_{\mu_0}(\eta_t(x))$ as the average occupancy of site $x$ over the sample paths that we generate,
which provides an estimator of $S$.

These three quantities $\hat V,D,S$ all converge as a function of time to stationary values, providing differing information as to
the rates of convergence.  Note that for non-interacting particles, $\rhobar(x)=\mathbb{E}_\pi(\eta(x))={\rm e}^{-V(x)}/z$ for
some constant $z$, so $D(t) = -S(t)+\hat{V}(t) + \log z$.

\subsubsection{One-dimensional case -- results}
\label{sec:1d}

{We consider periodic boundaries for a model on a one-dimensional strip, this is equivalent to motion
on the perimeter of a circle (flat torus in one dimension). 
In this case condition~\eqref{eqn:condition2} requires $k_{x,x+1} =k_{x-1,x}$, so we set $k_{x,x+1}=c$ with some constant $c$ that is independent of $x$. 
The choice $c>0$ corresponds to a fixed force $c\;\!\mathrm e^{V}$
that is forcing the particles to travel around the circle.  For a hydrodynamic limit consistent with macroscopic fluctuation theory, we require $c$ to vary with the system size $L$ as $c=E/L$ with $E$ a fixed constant~\cite{Bertini2015a}.

We note in passing that the use of periodic boundaries is essential for breaking balance in these closed systems: on a finite strip with reflecting boundary conditions, \eqref{eqn:condition2} has no solutions except $k_{x,y}=0$ so there is no way to break detailed balance.}

Thus, returning to the case with the periodic boundaries, the generator is
\begin{multline}
  \label{eqn:generator_1d}
  \mathcal L f(\eta) = \sum_{x=0}^{L-1} \Bigl[\bigl(f(\eta^{x,x+1})-f(\eta)\bigr)\;\!L^2g(\eta(x))\bigl( \mathrm e^{(V(x)-V(x+1))/2} 
  + {(E/L)}%c\;\!
  \mathrm e^{V(x)}\bigr)\\
  +\bigl(f(\eta^{x,x-1})-f(\eta)\bigr)\;\!L^2g(\eta(x))\bigl( \mathrm e^{(V(x)-V(x-1))/2} - {(E/L)}%c\;\!
  \mathrm e^{V(x)}\bigr)\Bigr],
\end{multline}
where the addition is periodically extended on $\Lambda_L=\{0,\dots,L-1\}$, i.e., $(L-1)+1=0$ and $0-1=L-1$. We take $g(k)=k$ so
that the particles do not interact.  The potential is
\begin{equation}\label{eqn:1dpotential}
  V(x)=A \sin(4\pi x/L)- B \cos(2\pi x/L)
\end{equation}
with $A=3/2$ and $B=3/4$ so that the global minimum of the potential is at $\hat x \approx 0.888$ with $V \approx -2.052$. The
height of the barrier is approx $2.609$.  The initial condition has all particles on a single site, $x_0=L/4$, in the vicinity of
the secondary minimum.  The stationary state has $\rhobar(x)=\mathbb{E}_\pi(\eta(x)) \propto {\rm e}^{-V(x)}$ with a
proportionality constant determined by the total density (which in this case is $z\approx 2\;\!377)$. {The parameter $E$ in
Eq.~\eqref{eqn:generator_1d} is set to $E=36$. For the lattice size $L=300$, the maximal value allowed for $E$ to ensure that $c_s+c_a\ge 0$ is slightly above $38.4$. In principle one can choose larger values for $E$ by increasing the lattice size $L$.

 The results in Fig.~\ref{img:1Dno1} are for a domain of size $L=300$; we also compared this to simulations for
  $L=150$, $L=300$ and $L=450$ for the value $E=18$ (to ensure positiveness of the transition rates for $L=150$). We found the results to be qualitatively very similar, see the bottom right panel in
  Fig.~\ref{img:1Dno2}.}  Fig.~\ref{img:1Dno1} shows the convergence to equilibrium of the mean potential energy and the entropy.
One sees that convergence of both the energy and the entropy is significantly faster when detailed balance is broken.  To
illustrate the mechanism for this effect, Fig.~\ref{img:1Dno2} shows how the mean density $\mathbb{E}_{\mu_0}(\eta_t(x))$ varies
with time.  In the irreversible case, the non-gradient part of the drift force $E$ acts to the right and is equal to
$c\;\!{\rm e}^V$, so it is large near the maxima of the potential.  This prevents the system from becoming localised in the
secondary (local) minimum and aids convergence to the steady state.  By contrast, in the reversible system, the particles need to
\emph{diffuse} over the maxima of the potential, which is a slower process.  This difference explains the much faster convergence to the
steady state observed in Fig.~\ref{img:1Dno1}.  The overshoot of the entropy for the reversible case in Fig.~\ref{img:1Dno1}
occurs because the state where the particles are distributed evenly between the two minima has a higher entropy $S$ than the
steady state (where they are localised primarily in the global minimum).  {The state where the particles are distributed evenly between the minima is
an example of a situation where the gradient of the free
energy is small (within the relevant metric), so that steepest descent of the free energy leads to slow changes in the density.}

\begin{figure}
  \center
  \includegraphics[width=5.8cm]{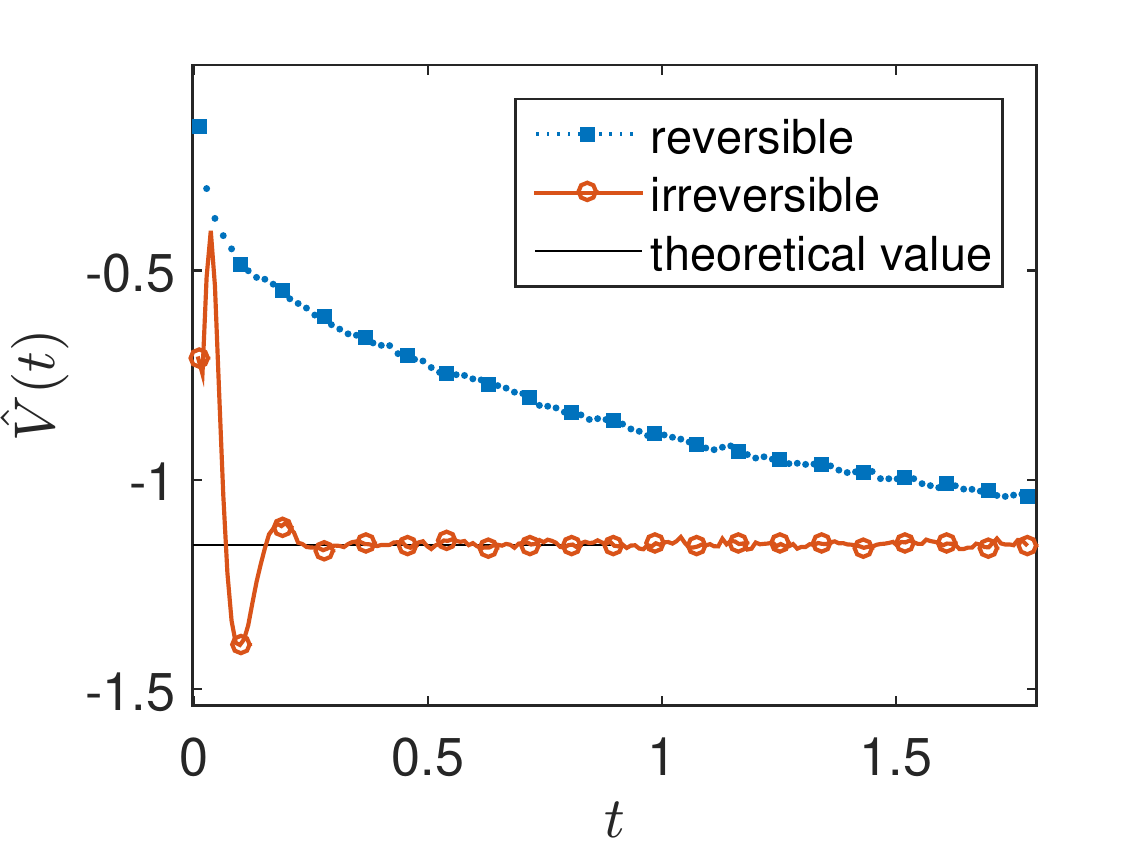}
  \includegraphics[width=5.8cm]{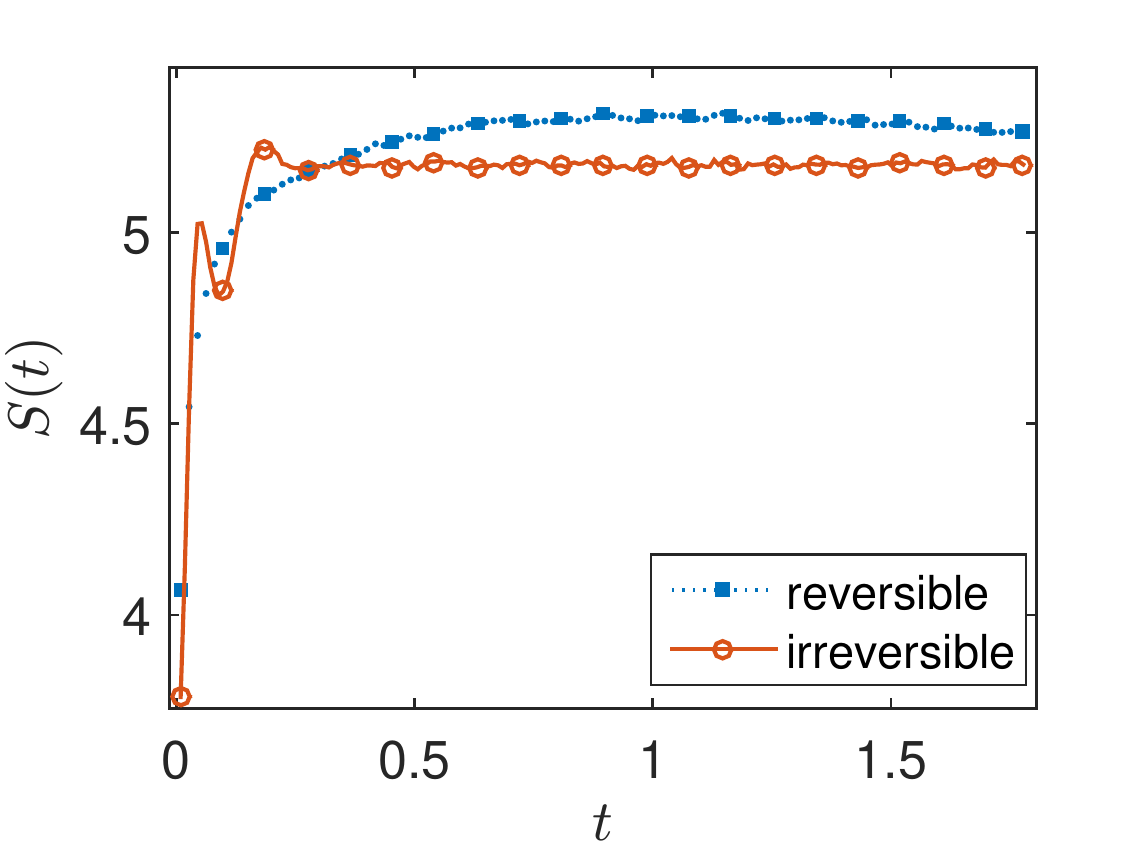}
  \includegraphics[width=5.8cm]{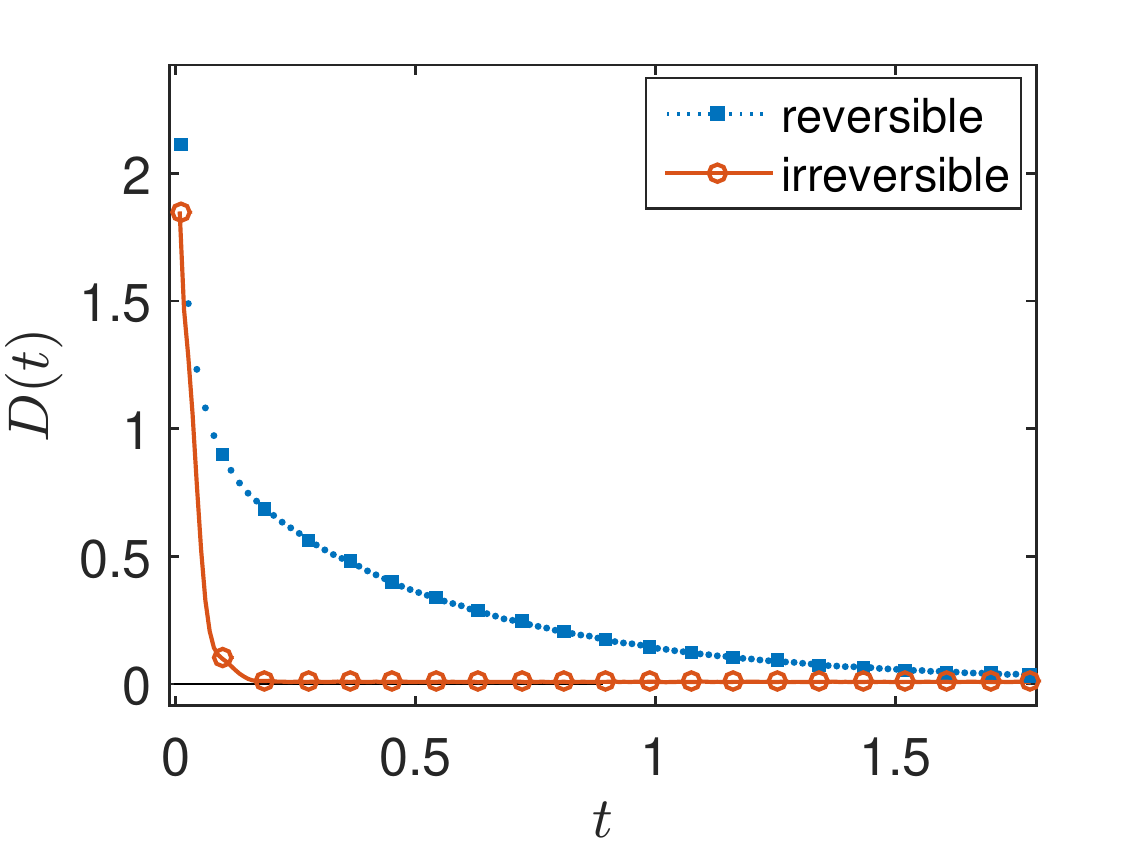}
  \includegraphics[width=5.9cm]{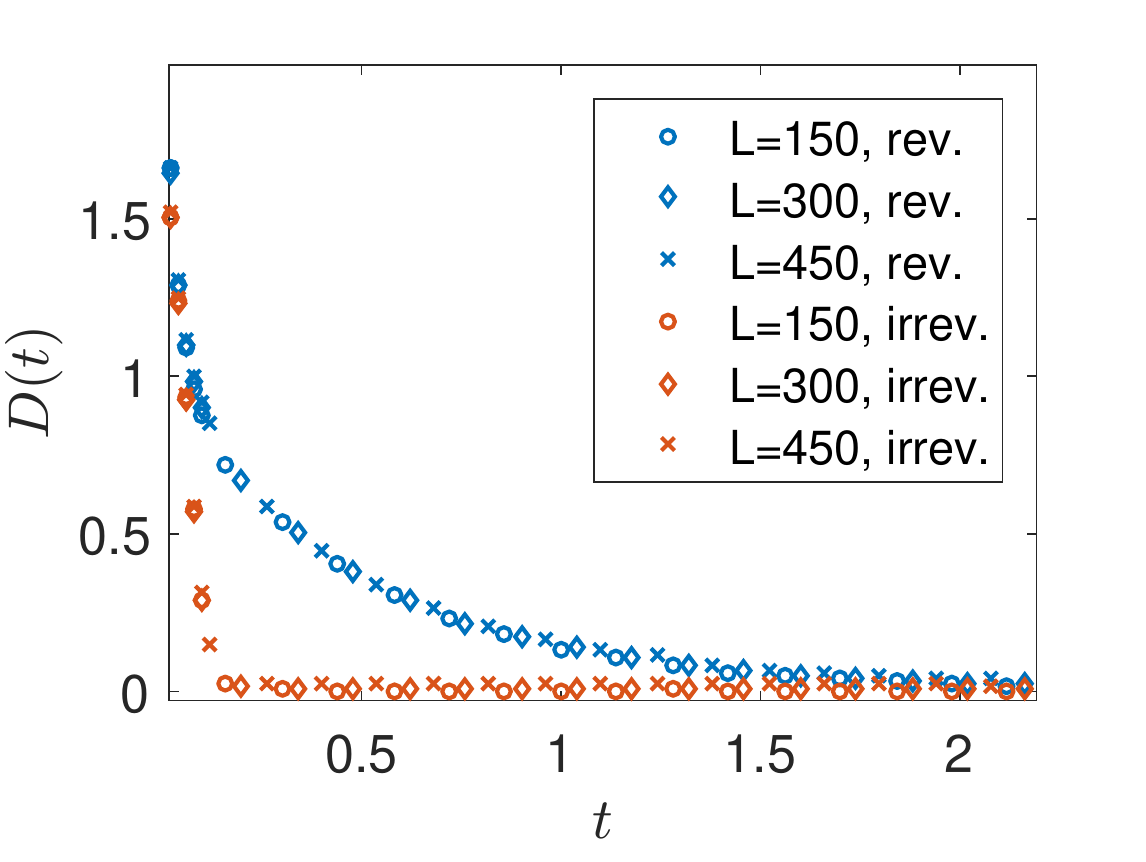}
  \begin{minipage}[b]{0.95\linewidth}
    \caption{\label{img:1Dno1}One-dimensional simulation for independent particles on a circle with $L=300$ sites, comparing
      reversible and irreversible drift-diffusion processes as described in the main text with the potential
      \eqref{eqn:1dpotential}.  Top row and bottom left: Plot of the test observables average energy $\hat{V}$, Gibbs entropy $S$
      and relative entropy $D$ {for $E=36$}.  %$c=0.12$. 
      { Bottom right: Plot of the relative entropy $D$ for different system
        sizes $L=150,300,450$, all for {$E=18$.  As predicted by the hydrodynamic equation, varying the system size at fixed $E$ and rescaling time by a factor of $L^2$ leads to limiting behaviour independent of $L$.}}
        %c=$180/L$ (i.e. $c=0.12, 0.06, 0.04$, respectively). } 
        All results were obtained by
      averaging over $20,000$ individual particle trajectories.}
  \end{minipage}
\end{figure}

\begin{figure}
  \center
  \hspace*{-60pt}\includegraphics[width=15.75cm]{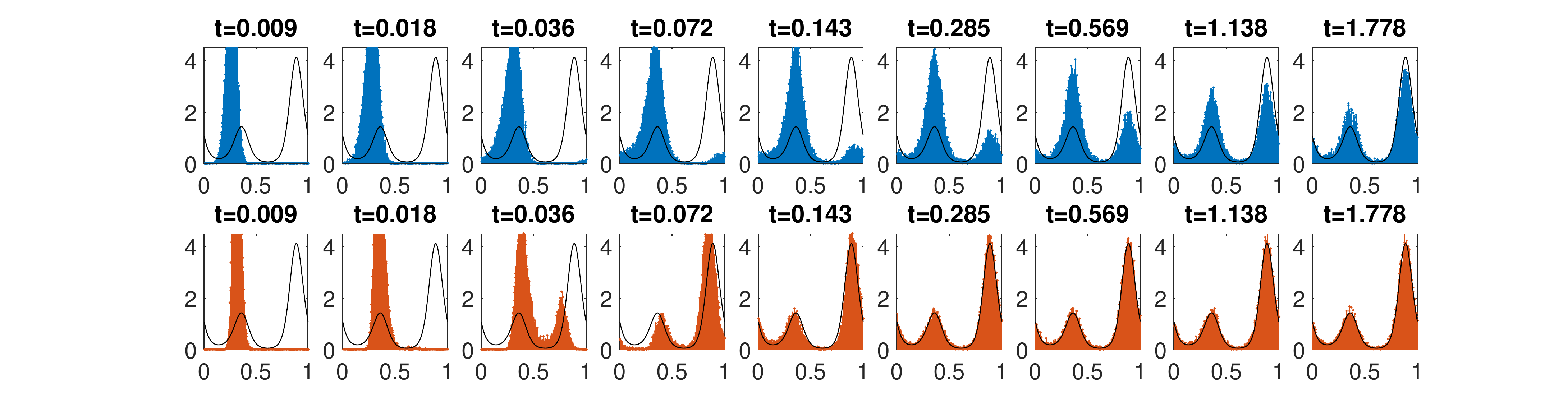}
  \begin{minipage}[b]{0.95\linewidth}
    \caption{\label{img:1Dno2}One-dimensional simulation for independent random walk on a circle with the potential
      \eqref{eqn:1dpotential}. Configuration at different times for the reversible (top row) and the irreversible (bottom row)
      process with drift `to the right' and steady state (in black). x-axis: Position. y-axis: Averaged number of particles.  {In the irreversible case, $E=36$.}}
  \end{minipage}
\end{figure}

{ Note also that \eqref{eqn:relent_000} implies that $D(t)\to0$ at long times, as the system converges to its steady
  state.  However, in Fig. \ref{img:1Dno1} one sees that our estimate of $D(t)$ converges instead to a small positive constant.
  This offset arises because our estimator of $D(t)$ is biased: it is based on $m$ independent numerical simulations (each with
  $N$ particles) and the expectation value of our estimator converges to $D(t)$ only as $m\to\infty$.  Specifically, we estimate
  $\mathbb E_{\mu_0}(\eta_t(x))$ as $\vartheta_t(x) = m^{-1}\sum_{k=1}^m\eta^k_t(x)$ where $\eta^k_t(x)$ is the number of
  particles on site $x$ at time $t$ in the $k$-th simulation.  Inserting this estimate into the (nonlinear) expression
  \eqref{eqn:relent_000}, it is easily shown that the resulting estimator of $D(t)$ has in general a finite bias. However, as
  $m\to\infty$, $\vartheta$ obeys a law of large numbers and converges almost surely to $\mathbb E_{\mu_0}(\eta_t(x))$ -- hence
  our estimator converges to $D(t)$ as $m\to\infty$.  }

\subsubsection{One-dimensional case -- discussion}

This one-dimensional model is useful for illustrative purposes and establishes the general principles derived in Sec.~\ref{sec:theory}.  
{%
However, the restriction to one dimension means that detailed balance can only be broken by applying a driving force $c\,{\rm e}^{V}$ (otherwise the invariant measure would be changed).   If barriers are large, one sees that the driving force near the top of the barrier must be very large indeed: it is hard to see how this can be realised in practical applications.  Physically, the idea is to drive a constant current around the periodic system, and this requires the drift velocities (and hence forces) to be largest at the top of any barriers, where the density is least.  In this sense, it is perhaps not surprising that by applying large forces to quickly drive particles over all barriers in the system, one can significantly speed up mixing of the particles between the two minima of the potential.

For these reasons, we turn to a two-dimensional system, where there are many more ways of breaking detailed balance while preserving the same invariant measure.
}

%However,
%the anti-symmetric rates $k_{x,y}$ depend on the absolute value of the potential $V(x)$ where one expects on physical grounds that
%properties of physical systems should depend only on potential gradients.  Also, in regions where the potential is large, one
%finds very large anti-symmetric contributions to the rates, which may be be inconvenient numerically and hard to achieve in
%physical systems.  However, this choice of anti-symmetric rates is the only possibility in one dimension.  For this reason, we now
%consider a two-dimensional system.

\subsubsection{Two dimensional case -- model and results}
\label{sec:2d}

In two dimensions, there is considerably more freedom in the choice of the rates $k_{x,y}$. If one again assumes periodic
boundaries, it is always possible to have all non-gradient forces acting in a single direction: for example $k_{x,x+e_1}=c$ where
$e_1$ is a lattice vector, as in the previous one-dimensional example.  {However, this requires driving forces that depend exponentially on the value of the potential, as in one dimension.  We therefore pursue a different strategy.}

\begin{figure}
  \center
  \includegraphics[width=7cm]{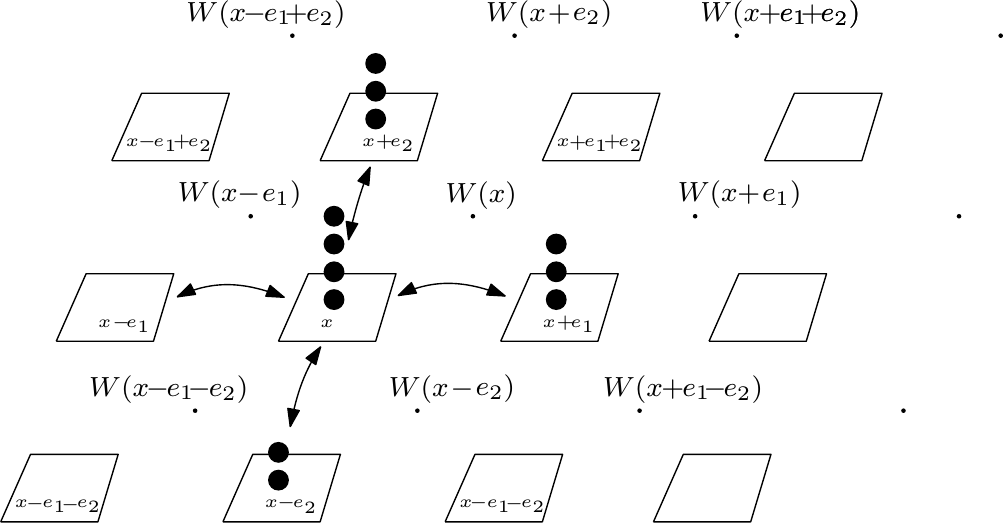}
  \caption{\label{fig:plaquettes}The function $W$ defined on the plaquettes; $e_1$ and $e_2$ are the Euclidean basis vectors.}
\end{figure}

Denoting the Euclidean basis for $\Lambda_L$ with $e_1,e_2$, Eq.~\eqref{eqn:condition2} implies that both
$k_{x,x\pm e_j}= -k_{x\pm e_j,x}$ and $k_{x,x+e_1}+k_{x,x-e_2}+k_{x,x-e_1}+k_{x,x+e_2} =0$ have to be satisfied. One way to choose
appropriate $k_{x,y}$ is to consider the plaquettes of the square lattice as in Fig.~\ref{fig:plaquettes} and to define a
\emph{vorticity} $W$ at the centre of each plaquette.  The value of $W$ on the plaquette centred at $x+\frac12(e_1+e_2)$ is
$W(x)$.  One then can choose the rates $k_{x,y}$ as the following differences:
\begin{align}
  k_{x,x+e_1}&= W(x\!-\!e_2)-W(x) \label{eq:plaquette}
  \\
  k_{x,x-e_2}&= W(x\!-\!e_1\!-\!e_2)-W(x\!-\!e_2)
               \notag  \\
  k_{x,x-e_1}&= W(x\!-\!e_1)-W(x\!-\!e_1\!-\!e_2)
               \notag \\
  k_{x,x+e_2}&= W(x)-W(x\!-\!e_1) \notag
\end{align}
This choice satisfies both conditions $k_{x,y}=-k_{y,x}$ and $\sum_y k_{x,y}=0$.  The quantity $W$ can be identified as a
vorticity, in the sense that taking $W(x)=W_0\delta_{x_0,x}$ with $W_0>0$ causes particles to circulate clockwise around plaquette
$x_0$.

\begin{figure}\label{img:potential}
  \center
  \hspace*{-38pt}\includegraphics[width=14cm]{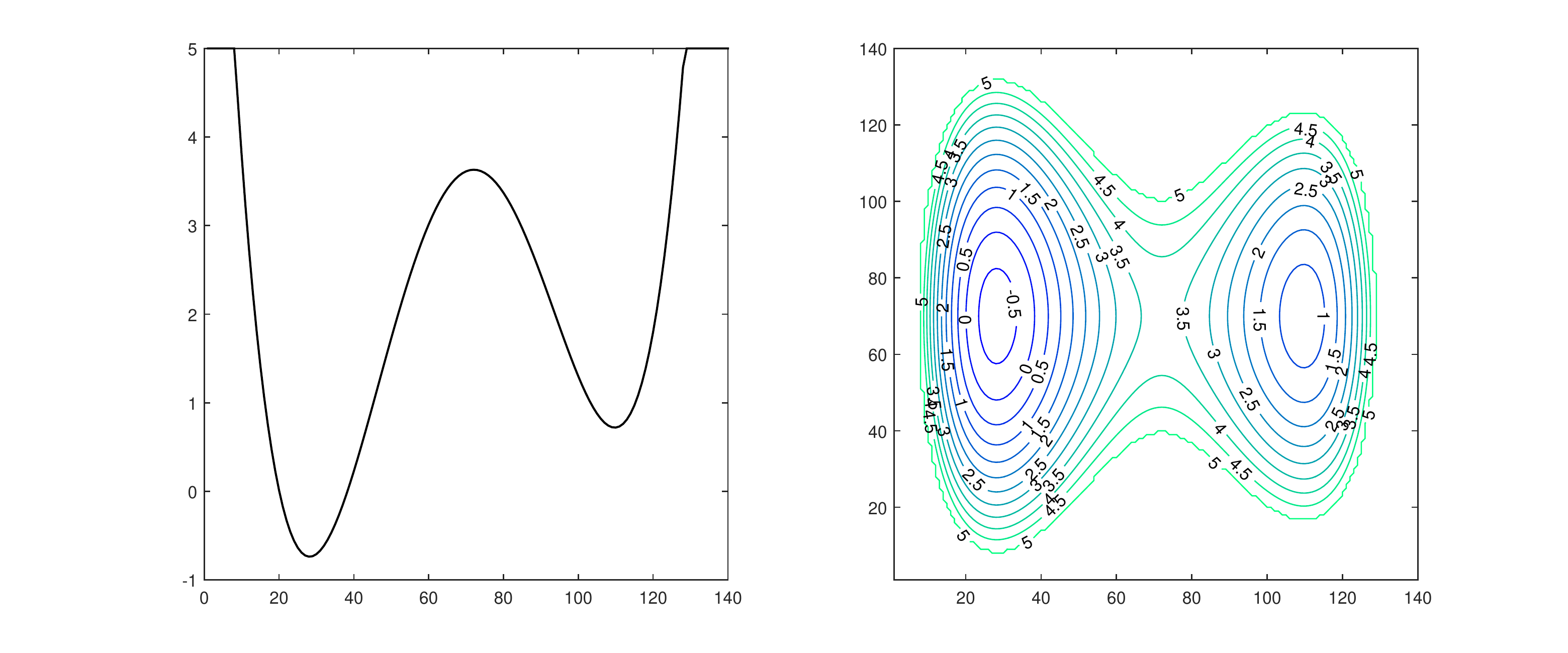}
  \begin{minipage}[b]{0.9\linewidth}
    \caption{\label{fig:levelset}Left: $x_1$-cross-section of $V(x)$ as given in \eqref{eqn:2dpotential} for $x_2=1/2$. Right: Level sets of $V(x)$.}
  \end{minipage}
\end{figure}

Any choice of the function $W$ is possible, and should lead to acceleration of the dynamics, following the theoretical analysis of
Sec.~\ref{sec:theory}.  Here we concentrate on a case where $W$ is related to the potential $V$, so that the rates $c(x\to y)$
depend only on the gradients of the potential in the vicinity of site $x$.  {(The physical idea is that particle motion is naturally sensitive to local potential gradients since these correspond to forces acting on the particles.  On the other hand, the motion of a particular particle should not be sensitive to the total energy $V$, since this depends on the state of the system far away from that particle.)}
To arrive at forces that depend only on potential gradients, we take
$W(x)=a\cdot \exp(\frac14[V(x)+V(x+e_1)+V(x+e_2)+V(x+e_1+e_2)])$, where $a$ is a parameter that sets the scale of the vorticity.

On taking the hydrodynamic limit, this gives rise to the driving force
\begin{equation}
  \label{eqn:drivingforce}
  E(\tilde x) = -\nabla \tilde V(\tilde x) + a[ e_1 \nabla_2 \tilde V(\tilde x) - e_2 \nabla_1 \tilde V(\tilde x) ] , 
\end{equation}
where $a>0$ (recall from Sec.~\ref{sec:Hydrodynamic-limit} that $\tilde x$ is the image in $\Lambda$ of the discrete site
$x\in \Lambda_L$).  We recognise the second term on the right hand side as a force that is obtained by rotating $\nabla V$
clockwise by $\pi/2$ radians, so that it acts to drive the system around the level sets of $V$.

\begin{figure}
  \center
  \hspace*{-1.8cm}\includegraphics[width=15cm]{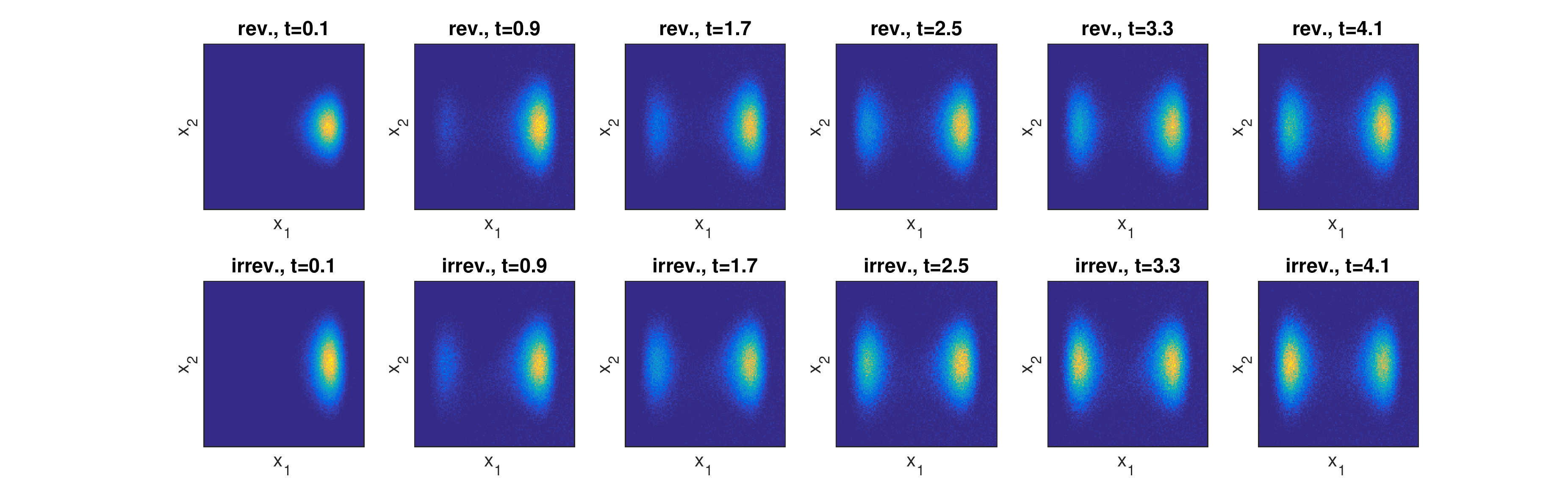}
  \begin{minipage}[b]{0.95\linewidth}
    \caption{\label{img:timeevolution1}Configuration for $g(k)=k$ with the potential \eqref{eqn:2dpotential} at different
      times. (Dark) blue means low number of particles, yellow means many particles. Top row: reversible process. Bottom row:
      irreversible process.}
  \end{minipage}
\end{figure}

\begin{figure}
  \center
  \includegraphics[width=5.8cm]{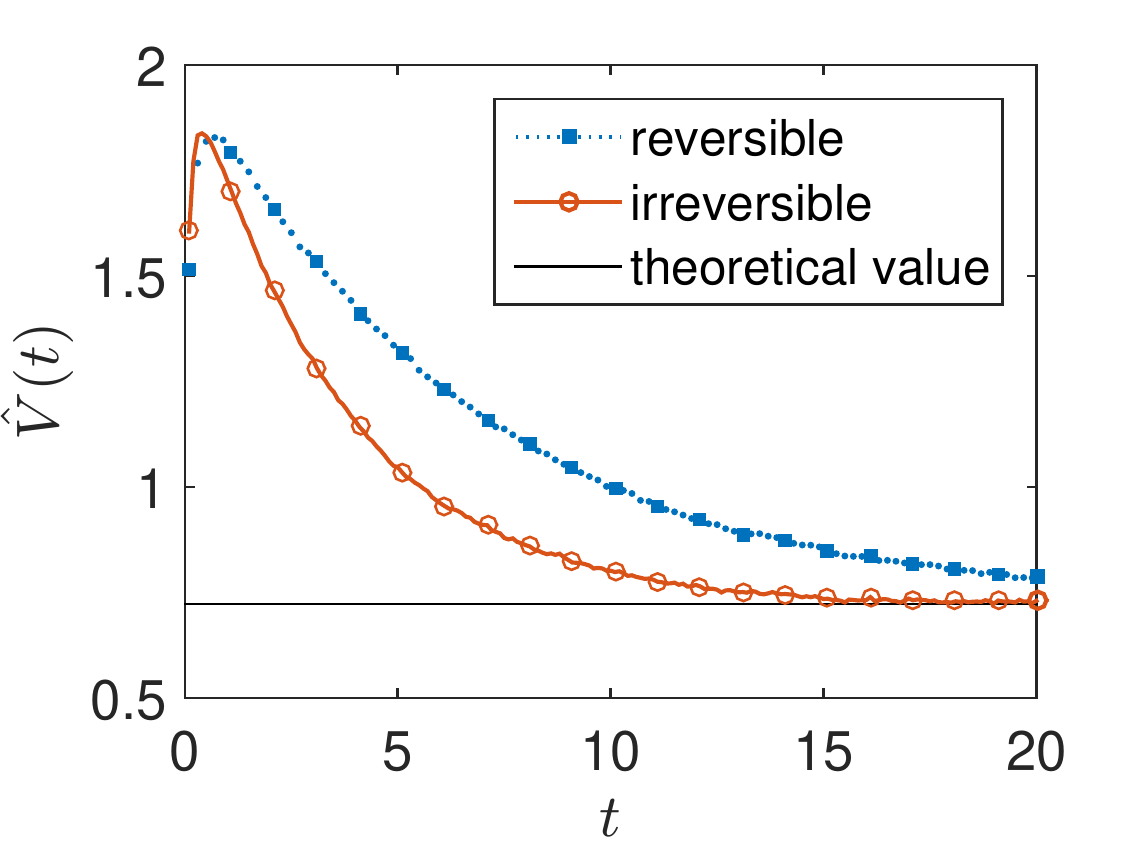}
  \includegraphics[width=5.8cm]{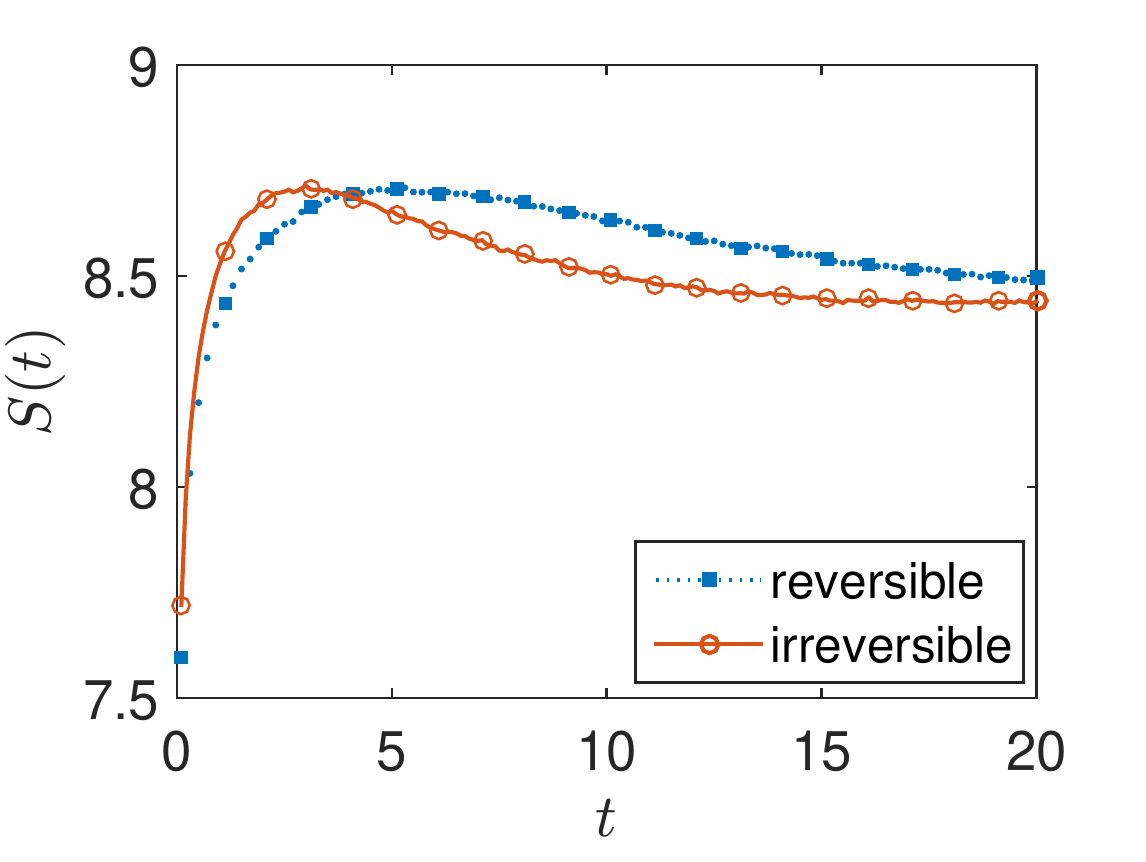}
  \includegraphics[width=5.8cm]{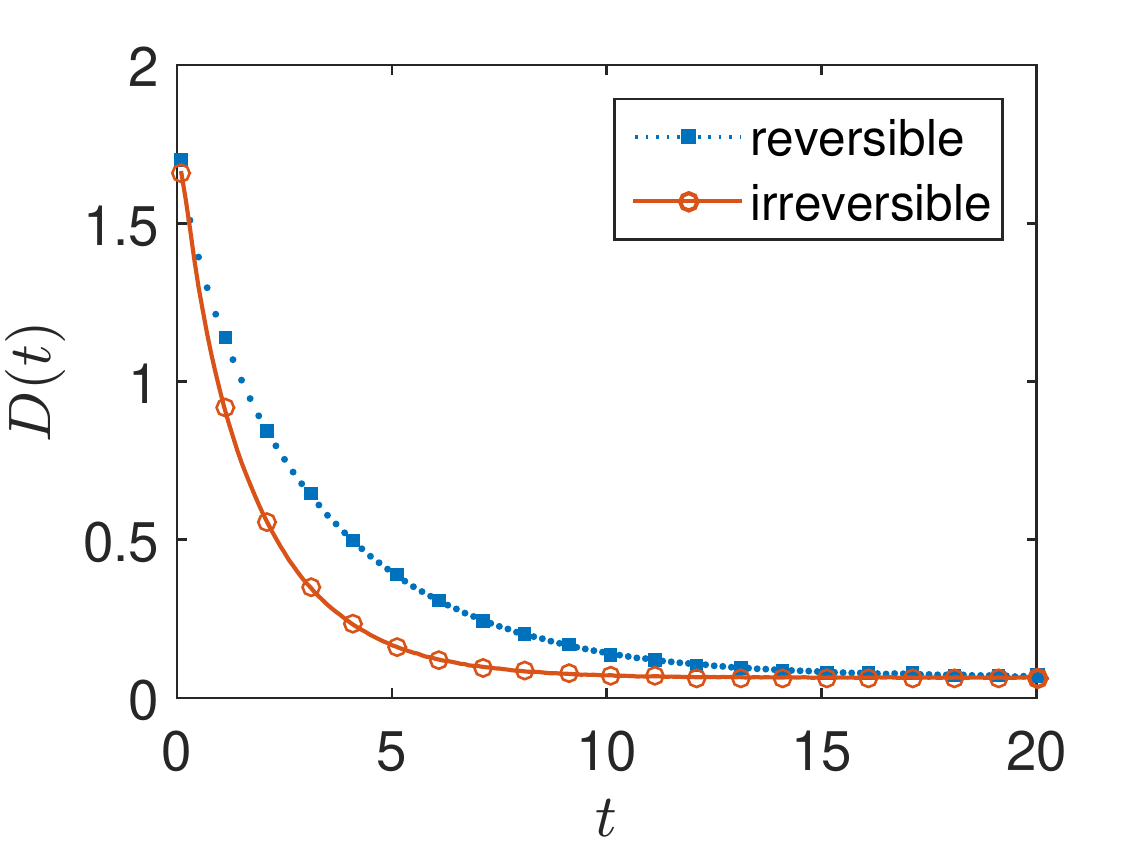}
  \includegraphics[width=5.8cm]{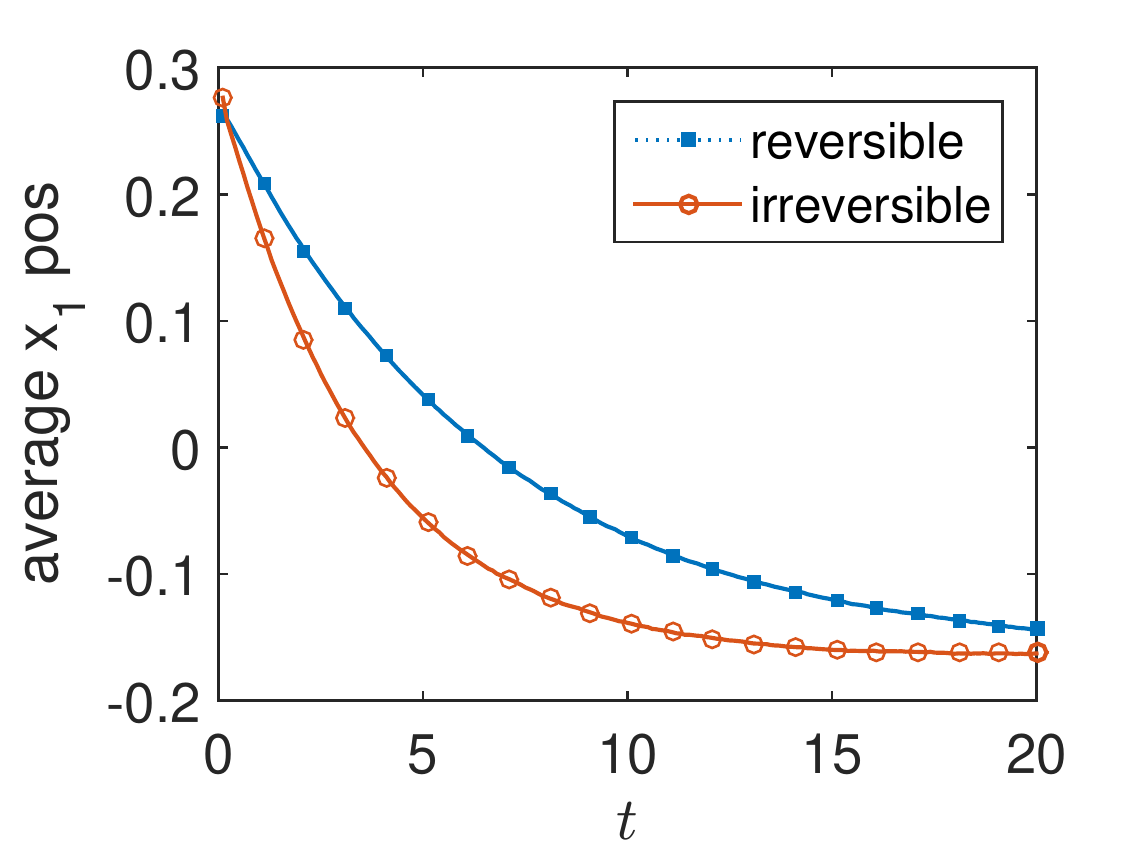}
  \begin{minipage}[b]{0.95\linewidth}
    \caption{\label{img:2dindependent} ZRP with $g(k)=k$ (independent random walk) with initial position of all particles in the
      local minimum. Average energy $\hat V$, the Gibbs entropy $S$, the relative entropy $D$ and the average
      $x_1$-position of particles. The initial position of the particle is at a fixed position in the local (but not global)
      minimum of the potential \eqref{eqn:2dpotential}. The domain size is $L^2=140^2$ and we averaged over $16$ simulations consisting of  $9800$ particles each.}
  \end{minipage}
\end{figure}

The following simulations are on a two dimensional closed domain with $L=140$ and zero flux at the boundary, i.e., the domain has
$140\times 140=19\;\!600$ sites and the particles cannot leave the domain.

\begin{figure}
  \center
  \includegraphics[width=5.8cm]{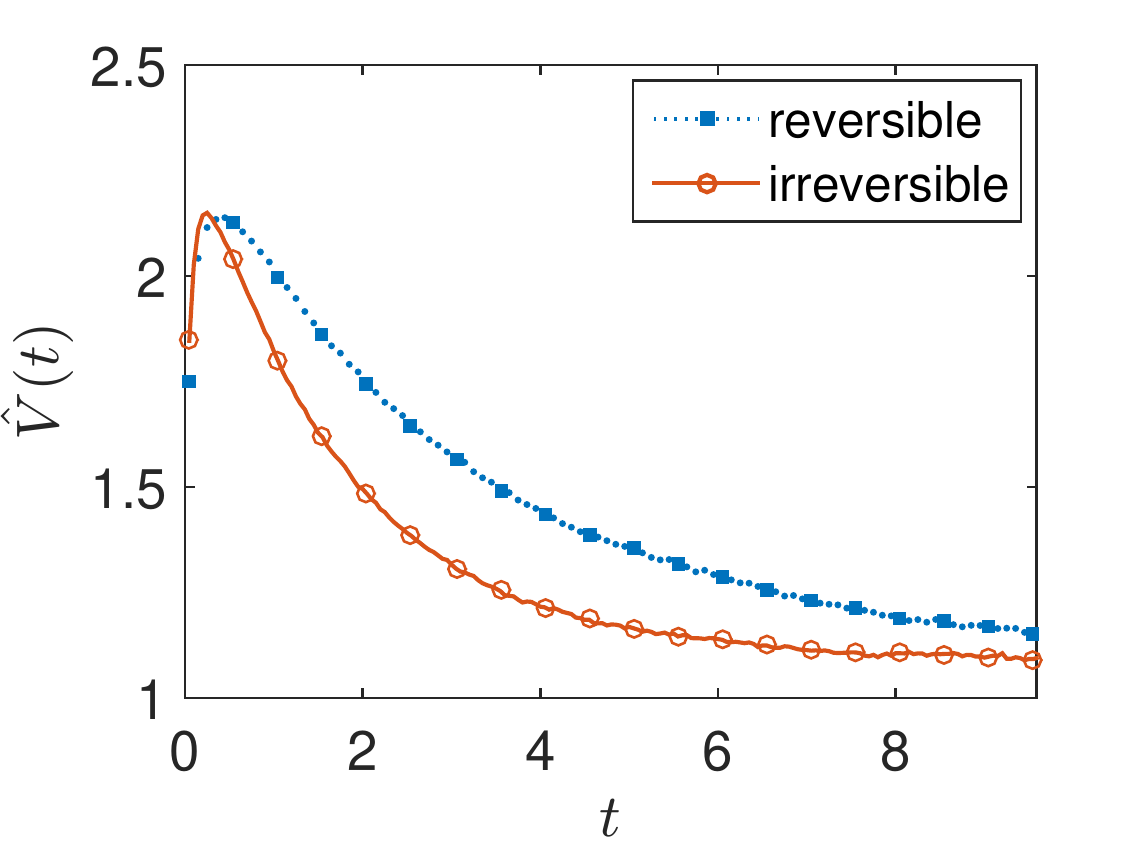}
  \includegraphics[width=5.8cm]{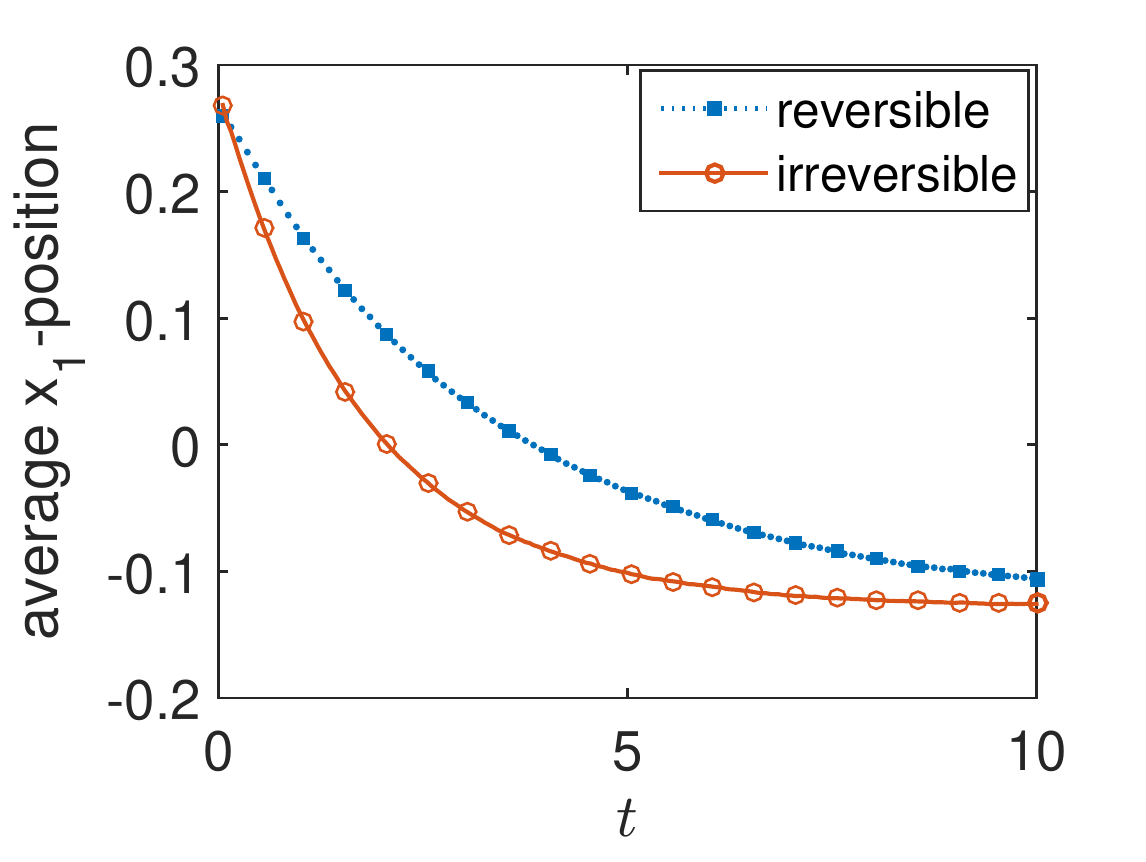}
  \begin{minipage}[b]{0.95\linewidth}
    \caption{\label{fig:ZRPobservables_super}ZRP with $g(k)=k^{3/2}$ and the particles are started in the local minimum. Left:
      Average energy $\hat V$. Right: average $x_1$-position.}
  \end{minipage}
\end{figure}

\begin{figure}
  \center
  \includegraphics[width=5.8cm]{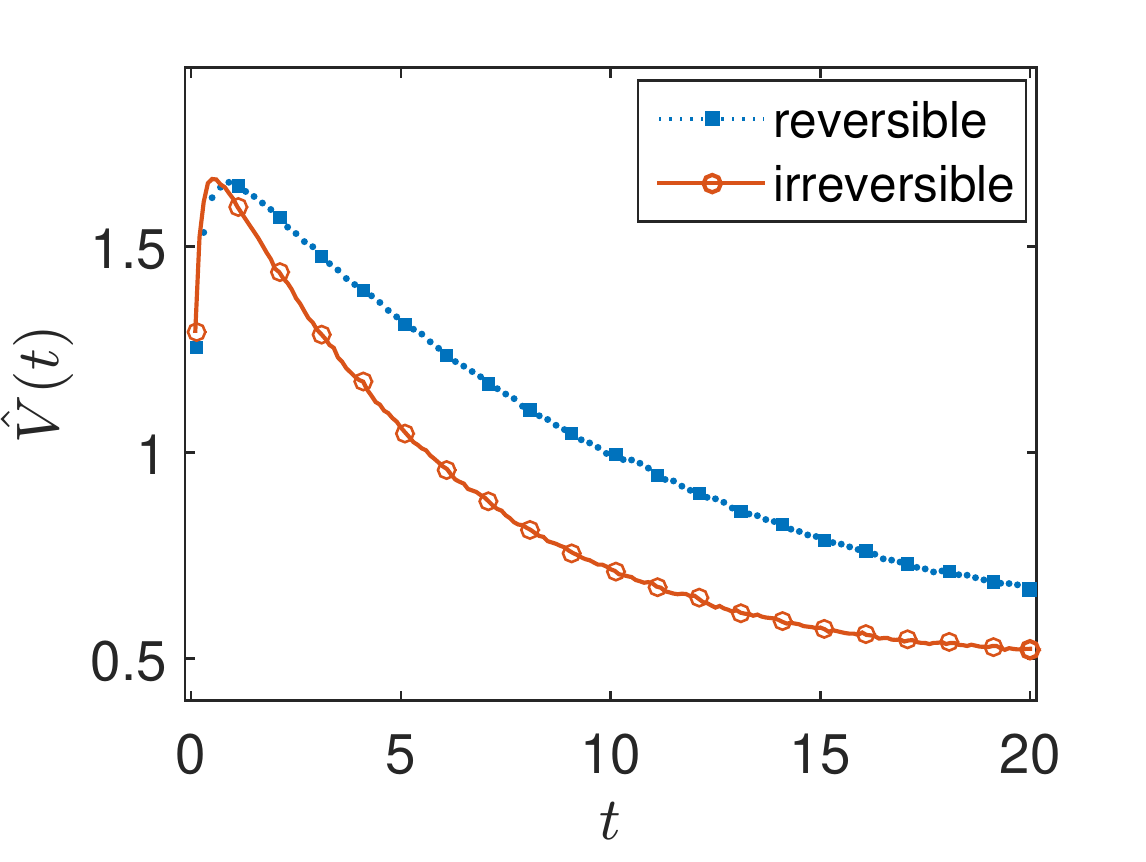}
  \includegraphics[width=5.8cm]{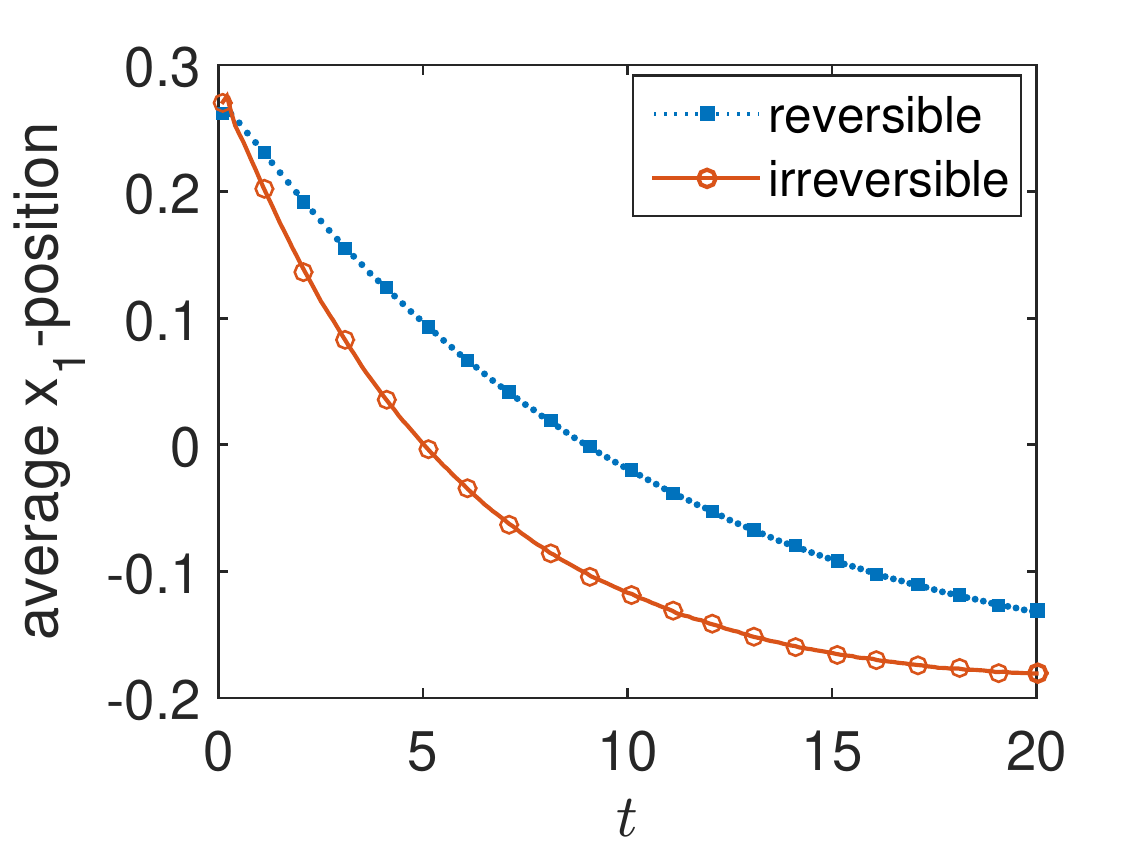}
  \begin{minipage}[b]{0.95\linewidth}
    \caption{\label{fig:ZRPobservables_sub}ZRP with $g(k)=k^{5/6}$ and the particles are started in the local minimum. We again plot the average energy $\hat V$ and the average $x_1$-position.}
  \end{minipage}
\end{figure}
 
We consider three different ZRPs, corresponding to different choices for $g(k)$. Firstly, we consider the linear case (independent
particles), where $g(k)=k$. We further consider the superlinear case with $g(k)=k^{3/2}$, such that the particles repel each other
(the hop rates away from site $x$ is increased when that site contains more particles). Finally we investigate the sublinear case
with $g(k) = k^{5/6}$ in which the particles prefer to cluster together. {For each setting, we simulated the process with both
reversible and irreversible dynamics, with $L^2/2=9\;\!800$ particles averaged over 16
simulations.} The potential, which is also depicted in Figure~\ref{fig:levelset}, is for shifted coordinates
$x=(x_1,x_2)\in[-1/2,1/2]^2$ given by
\begin{equation}\label{eqn:2dpotential}
  V(x_1,x_2) = A(x_1^2 - B)^2 + C x_2^2 + D x_1
\end{equation}
with a cut-off at a given height {$V^*$. For the simulations we chose the parameters $A=500$, $B=0.085$, $C=30$, $D=2.5$ and $V^*=5$ (that is, the potential used is ${\rm max}(V(x_1,x_2),V^*)$).} The parameter in~\eqref{eqn:drivingforce}, which sets
the strength of the non-gradient term of the driving force, was set to $a=0.4$. {This value is again close to the maximal allowed value (which is slightly above $0.405$).}

For all simulations, the particles start at position $(0.5,0.75)\in[0,1]^2$ close to the local minimum of the double well
potential. The particles then try to leave this well and move to the global minimum (on the left) as can be seen in the plots in
Fig.~\ref{img:timeevolution1} for the linear case. The test observables for the linear/superlinear/sublinear case can be found
in Figs.~\ref{img:2dindependent}, \ref{fig:ZRPobservables_super} and \ref{fig:ZRPobservables_sub}, respectively. Depending on the
chosen configuration, the simulation time on a HPC node with 16 cores using Matlab took between 10 and 13.5 hours.

As in the one-dimensional case, the particles are under the irreversible dynamics able to leave the minimum faster than it is the
case for reversible dynamics (compare the bottom row with the top row in Fig.~\ref{img:timeevolution1}).

\begin{table}
  \center
  \begin{tabular}{ l |l |l}
    $t_s$ & $\Delta\hat V$ & $\Delta x_1$ \\
    \hline
    $g(k) = k$ & 2.38  &  1.30 \\
    $g(k) = k^{3/2}$ & 1.11 & 0.58 \\
    $g(k) = k^{5/6}$ & 3.55 &  2.03\\
  \end{tabular}
  \hspace{36pt}
  \begin{tabular}{ l |l |l}
    $t_s/t_a$ & $\Delta\hat V$ & $\Delta x_1$ \\
    \hline
    $g(k) = k$ & 1.83  & 1.79 \\
    $g(k) = k^{3/2}$ & 1.78  & 1.86 \\
    $g(k) = k^{5/6}$ & 1.77 & 1.80 \\
  \end{tabular}
  \caption{Table of the absolute times $t_s$ for the reversible process (left) and ratios between times of the reversible and 
    irreversible process $t_s/t_a$ (right) to reach the distances $\Delta\hat V=0.3$ and $\Delta x_1=0.2$, respectively.}
  \label{table:times}
\end{table}

\subsubsection{Two dimensional case -- discussion}

{We close this section with Table~\ref{table:times}, which quantifies the acceleration in the models where particles attract, repel, or have no interactions.} For this, we consider
the average energy $\hat V$ and the average $x_1$ position of the particles. Assuming that the final values of these observables in  the irreversible
simulations are close to their steady-state values, we consider the distance $\Delta \hat V$ (resp. $\Delta x_1$) of both the
reversible and irreversible process and keep track of the first time where the distance is below a given threshold. Denoting this
time for the reversible process with $t_s$ and for the irreversible process with $t_a$, we can use the ratio $t_s/t_a$ as an
estimator for the acceleration.

From the data in the table, on sees that the processes are typically accelerated by factors about $1.75$ independent of the choice
of $g(k)$. We checked different thresholds (here we displayed $\Delta\hat V=0.3$ and $\Delta x_1=0.2$) which all 
lead to the same conclusions.  

{These are significant accelerations,  although considerably less than the dramatic speedup of order $10$ observed in one dimension.  However, the physical mechanisms for the acceleration are different in the two cases.  In one dimension, the drift forces which act to push particles up and over the barrier, so the forces are very large at the top of the barrier.  In two dimensions, the effect is more subtle: returning to Fig.~\ref{fig:levelset} and recalling that the drift force in (\ref{eqn:drivingforce}) is obtained by a  rotation of the potential gradient, one sees that in the vicinity of the saddle point of the potential, there is a net drift to the left in the top part of Fig.~\ref{fig:levelset}(b), and a drift to the right in the bottom part.  A natural analogy is a gentle stirring motion that happens in the vicinity of the saddle point, and tends to accelerate mixing.  This seems a much more plausible mechanism for accelerating convergence to equilibrium in practical situations, compared with the large forces required in one dimension.

Finally, we note that transport between the minima of a non-convex potential energy always involves a slow time scale proportional to ${\rm e}^{\Delta V}$, since a particle must still reach the barrier in order to cross it, and the probability that a particle visits the barrier is proportional to ${\rm e}^{-\Delta V}$.  However, the results here show that mixing of particles between energy minima can be accelerated by enhancing the probability that if a particle reaches a region with high $V$, it takes advantage of this excursion in order to cross the barrier.  The mechanisms for this enhanced probability differ between the models considered here -- it would be interesting to investigate this effect further, so as to understand how general these mechanisms are and how they can be exploited in practical applications.%
}

\section{Conclusion and Outlook}
\label{sec:conc}

We have considered interacting particle systems described by Markov chains, and their hydrodynamic limits, as described by
macroscopic fluctuation theory.  We compare reversible and irreversible processes: for an irreversible system with generator
$\LL$, the corresponding reversible process is the one identified in~\eqref{equ:LSA}, whose generator is $\LL_S$.  At the
microscopic level, it is known that the irreversible process then converges to its steady state at least as fast as the reversible
one -- this can be demonstrated by considering either the spectral gap or the (level-2) large deviations of the empirical measure.
In the hydrodynamic limit, Eq.~\eqref{equ:i2-faster-mft} shows that this property is preserved, by considering the large
deviations of the empirical density.  Moreover, Eq.~\eqref{eqn:LDPfunctional2} gives a quantitative expression for the
acceleration of convergence, which may be seen as a generalisation of previous results for single-particle
diffusions~\cite{Rey-Bellet2015a}.

Our numerical results for the ZRP reinforce the observation that for a given reversible system, there is a large family of
irreversible systems for which convergence to equilibrium is faster (or, at least, equally fast).  We considered two cases: either
a drift force in a single direction, which acts to drive a system around a circle (Sec.~\ref{sec:1d}) or the introduction of a
force that drives the system around the level sets of the potential (Sec.~\ref{sec:2d}).  In both cases, we observe acceleration
of convergence, as expected.

The results within MFT provide a geometrical interpretation of the acceleration, in terms of forces that act in directions
perpendicular to the free energy gradient, as shown by orthogonality relations for currents such as Eq.~\eqref{equ:JAdecomp}.  We
have argued that such forces can act to accelerate convergence by driving the system away from regions where the free energy
gradient is shallow, in which cases reversible processes exhibit slow convergence.
 
We offer two perspectives on future application of these ideas. First, we have shown that breaking detailed balance generically
accelerates convergence, but of course there are very many ways to write down irreversible models, and it is not clear what choices
are most practical in applications, nor which ones lead to the fastest convergence.  In particular, the choice considered for ZRP
examples shown here are rather specific to systems in one or two dimensions. (We emphasise however that the configuration spaces
of the ZRP are very high-dimensional since we consider $N$ interacting particles, so the methods are not restricted to systems
with low-dimensional configuration spaces.)  Second, we gave a geometrical interpretation in which the symmetric dynamics correspond to the
gradient flow (steepest descent) of the free energy and the antisymmetric dynamics are in some sense orthogonal to this gradient
flow. This offers a potentially new perspective on hydrodynamic limits in irreversible systems, which it would be interesting to
investigate further, for example with a view towards obtaining analytic estimates for the rate of convergence. 

\vspace{4pt}
Supporting data for this manuscript and the code used for the simulations will be made available short after publication on the University of Bath data archive [DOI to be added].\\

{\bf Acknowledgements:}
  MK is supported by a scholarship from the EPSRC Centre for Doctoral Training in Statistical Applied Mathematics at Bath (SAMBa),
  under the project EP/L015684/1.  JZ gratefully acknowledges funding by the EPSRC through project EP/ K027743/1, the Leverhulme
  Trust (RPG-2013-261) and a Royal Society Wolfson Research Merit Award. This research made use of the Balena High Performance
  Computing (HPC) Service at the University of Bath. { The authors thank the anonymous referees for very helpful comments
    and suggestions.}

%% end paste

\bibliographystyle{spmpsci-jz}      % mathematics and physical sciences
\def\cprime{$'$} \def\cprime{$'$} \def\cprime{$'$}
  \def\polhk#1{\setbox0=\hbox{#1}{\ooalign{\hidewidth
  \lower1.5ex\hbox{`}\hidewidth\crcr\unhbox0}}} \def\cprime{$'$}
  \def\cprime{$'$}

\end{document}